\documentclass{article}

\usepackage[english]{babel}

\usepackage[letterpaper,top=2cm,bottom=2cm,left=3cm,right=3cm,marginparwidth=1.75cm]{geometry}

\usepackage{amsmath}
\usepackage{graphicx}
\usepackage{latexsym}

\usepackage{algorithm}
\usepackage{algpseudocode}
\usepackage{amssymb}
\usepackage{amsmath}
\usepackage{authblk}
\usepackage{booktabs}
\usepackage[footnotesize]{caption}
\usepackage{enumitem}
\usepackage{graphicx}
\usepackage[colorlinks=true, allcolors=blue, hidelinks]{hyperref}
\usepackage[switch]{lineno}
\usepackage{mathtools}
\usepackage{mathrsfs}
\usepackage{soul}
\usepackage{url}
\usepackage{xcolor}

\newtheorem{example}{Example}
\newtheorem{theorem}{Theorem}
\newtheorem{lemma}[theorem]{Lemma}

\newtheorem{definition}[theorem]{Definition}

\newcommand{\AW}{W_{\mathsf{A}}}
\newcommand{\M}{\mathcal{M}}

\newcommand{\opt}{\mathrm{opt}}
\newcommand{\PWin}{P}
\newcommand{\R}{\mathbb{R}}
\newcommand{\RV}{\mathsf{RV}}
\newcommand{\RW}{W_{\mathsf{R}}}
\newcommand{\SW}{\mathsf{SW}}
\newcommand*{\QED}{\null\nobreak\hfill\ensuremath{\blacksquare}}%
\newcommand{\wopt}{\mathrm{wopt}}

\newenvironment{proof}{{\noindent \it Proof.}}{\hfill $\blacksquare$\par}

\title{Multi-unit Auction over a Social Network}
\author[1]{Fang Yuan}
\author[1]{Mengxiao Zhang\thanks{Corresponding author. Email: mengxiao.zhang@uestc.edu.cn}}
\author[2]{Jiamou Liu}
\author[1]{Bakh Khoussainov}
\author[1]{Mingyu Xiao}
\affil[1]{School of Computer Science and Engineering, University of Electronic Science and Technology of China, China}
\affil[2]{School of Computer Science, The University of Auckland, New Zealand}
\date{}

\begin{document}
\maketitle

\begin{abstract}
Diffusion auction is an emerging business model where a seller aims to incentivise buyers in a social network to diffuse the auction information thereby attracting potential buyers. We focus on designing mechanisms for multi-unit diffusion auctions. Despite numerous attempts at this problem, existing mechanisms either fail to be incentive compatible (IC) or achieve only an unsatisfactory level of social welfare (SW). Here, we propose a novel graph exploration technique to realise multi-item diffusion auction. This technique ensures that potential competition among buyers stay ``localised'' so as to facilitate truthful bidding. Using this technique, we design multi-unit diffusion auction mechanisms MUDAN and MUDAN-$m$. Both mechanisms satisfy, among other properties, IC and $1/m$-weak efficiency. We also show that they achieve optimal social welfare for the class of rewardless diffusion auctions. While MUDAN addresses the bottleneck case when each buyer demands only a single item, MUDAN-$m$ handles the more general, multi-demand setting. We further demonstrate that these mechanisms achieve near-optimal social welfare through experiments.
\end{abstract}

\section{Introduction}
Online social networks such as Tiktok, Twitter, and Temu not only enhance our social connectivity, but also provide new business opportunities: A user is able to act as a seller on an online social network, launching sales campaigns through the virtual space \cite{park2012social,zhang2017online}.
Unlike traditional campaigns, a seller in this virtual market could leverage the social network to diffuse information.
By implementing an appropriate marketing strategy, sales information passed to only a few initial individuals may trigger widespread dissemination, reaching a large cohort of potential buyers. Efforts have thus focused on designing mechanisms, termed {\em diffusion auctions}, that incentivise buyers to reveal not only their hidden valuations, but also their social connections to that sales information diffuses in the network \cite{guoemerging,roughgarden2016twenty}.

Diffusion auction differs from traditional auction designs in many aspects.
First, classical tools such as Myerson's lemma no longer apply to incentive compatibility (IC) when the buyers are allowed to strategically declare both their valuations and social connections \cite{guoemerging}.
Then, even though the generic VCG mechanism can be conveniently extended to a diffusion auction, extreme cases exist that result in a large negative revenue for the seller \cite{li2017mechanism}.
Last, unlike the traditional auction designs, for diffusion auction no mechanism would simultaneously satisfy IC, individual rationality (IR), non-deficit (ND), as well as optimal social welfare \cite{takanashi2019efficiency}.
The fundamental challenge in designing diffusion auctions is to mitigate the intrinsic conflict between the seller's desire to attract more participants to the auction, and buyers' wish to lower competition. Namely, by diffusing auction information to neighbours, a buyer may increase the chance of being out-bidded by others as more buyers may join the auction.
Hence new ideas and  tools must be developed for designing diffusion auctions.

In recent years, numerous studies have proposed diffusion mechanisms for {\em single-unit auction}, i.e., where the seller has only one item to sell
\cite{li2018customer,li2019diffusion,zhang2020incentivize,zhang2020redistribution}.
For example,  the IDM mechanism -- one of the starting points of this field \cite{li2017mechanism} -- achieves IC using the notion of {\em  critical buyers}, individuals who have the ability to alter the level of competition, and rewarding  critical buyers for their losses due to information diffusion.
However, these mechanisms do not have guarantee on efficiency. In contrast, GRP mechanism \cite{lee16mechanism} achieves efficiency under a weakened form of IC.
Later, FDM \cite{zhang2020redistribution} and NRM \cite{zhang2020incentivize} focus on redistribution issues while layered and recursive DPDMs \cite{jia2023incentivising} consider privacy issues.
Moving beyond single-unit case, {\em multi-unit auctions} study cases when the seller has multiple (homogeneous) items to sell.
One would hope that this case, being a natural generalisation of the single-unit counterpart, could be addressed using  mechanisms similar to IDM.
However, repeated attempts have failed to satisfy the crucial IC property:
(1) The GIDM mechanism, proposed in \cite{zhao2018selling}, determines the allocation of items and rewards using critical buyers.
This mechanism, however, was pointed out to violate the IC property \cite{takanashi2019efficiency}. See App. A in our full paper. 
(2) The subsequent DNA-MU mechanism, proposed in \cite{kawasaki2020strategy}, also utilises critical buyers while further leveraging a priority order based on buyers' distances from the seller. Unfortunately, this mechanism is once again shown to be not IC \cite{guo2022combinatorial}. See App. B. These failures attest the importance and difficulty of finding a truthful multi-unit diffusion mechanism. 

Remarkably, two recent mechanisms for multi-unit diffusion auction have claimed to be truthful. First, the  SNCA mechanism \cite{xiao2022multi} extends the classical clinching auction to the social network context. Similar to DNA-MU, the mechanism grants a buyer who is closer to the seller in the social network a higher priority when determining the allocation of items.
This mechanism, however, relies on the buyers' budgets which is not available in the standard multi-unit diffusion auction. Then, the LDM-Tree mechanism \cite{liu2022diffusion} applies a layer-based iterative allocation process, 
where buyers in the same layer have equal distance to the seller.
Essentially, the LDM-Tree alleviates competition by restricting information diffusion. This significantly limits the achieved social welfare.
%
%
Moreover, as the algorithm uses certain feature about the social network which may not be known {\em a priori}, the mechanism cannot be applied to the general setting of diffusion auction. See App. C.
All these suggest that {\em the problem of designing reasonable mechanisms for multi-unit diffusion auctions
is far from being settled}.

\smallskip

\noindent {\bf Contribution.} In this paper, we focus on designing a reasonable multi-unit diffusion auction.
{\bf First,} we introduce a technique for realising multi-unit diffusion auction which iteratively explores the social network starting from the seller $s$.
At each iteration, a part of the network is {\em explored}. The mechanism chooses a {\em winner} from the explored buyers while {\em exhausting} some other buyers. The winner and exhausted buyers are incentivised to diffuse the auction information allowing more buyers to be explored.
 We design a mechanism, named MUDAN, that can be embedded into this framework. MUDAN allocates an item to winners during the graph exploration ``on the fly''. It is IC, IR, ND, non-wasteful (NW), while satisfying a weakened version of social welfare optimisation. Moreover, the optimality ratio of MUDAN is tight for any truthful diffusion auction that incentivises buyers without using reward. See Section~\ref{sec:mudan}. 
{\bf Then,} as our mechanisms are defined for {\em single-demand} multi-unit diffusion auction, where each buyer is assumed to demand only one item, we extend the study to {\em multi-demand} multi-unit diffusion auction. We present a reduction from multi-demand multi-unit diffusion auction to the single-demand counterpart. Thus MUDAN can be generalised to the multi-demand case and satisfy all mentioned properties. See Section~\ref{sec:multi}.
{\bf Last, } since several priority scores of buyers -- a key ingredient of the mechanism -- may be defined using a number of traversal schemes, our focus is on evaluating the effect of these traversal schemes to the auction outcomes. In particular, we propose {\em new-agent-based} selection and compare it against other possible schemes over three real-world social networks. We demonstrate that under reasonable valuation models, (i) our mechanism significantly outperforms the benchmarks achieving near-optimal social welfare and revenue, and (ii) new-agent traversal achieves the highest performance in terms of both metrics among all tested traversal schemes. See Section~\ref{sec:exp}.
We summarise the highlights of our achievements:
\begin{itemize}[leftmargin=*]
    \item A technique for realising diffusion auction that iteratively explores the social network,
     and a new truthful multi-unit diffusion auction MUDAN for the single-unit case.
    \item A reduction from multi- to single-demand multi-unit diffusion auction which preserves the mechanism properties and a multi-unit diffusion auction MUDAN-$m$ for the multi-unit case.

    \item Experimental result demonstrating that MUDAN-$m$ achieves near-optimal social welfare and revenue with the new-agent traversal scheme.
\end{itemize}

\section{Model and problem formulation}\label{sec:model}

We present our model for {\em single-demand multi-unit} diffusion auction which was addressed by GIDM \cite{zhao2018selling} and DNA-MU \cite{kawasaki2020strategy}. The more general case of {\em multi-demand multi-unit} diffusion auction will be discussed in Section~\ref{sec:multi}.
Our model consists of the following:
\begin{itemize}[leftmargin=*]
    \item a seller, $s$, has $m\geq 1$ homogeneous items to sell.
    \item $n$ buyers $B=\{1,2,\ldots,n\}$. Each buyer $i\in B$ demands one item and attaches a valuation $v_i\in \R_+$. We often call buyers or the seller the {\em agents} of the network.
    \item a social network, represented as a directed graph $G=(B\cup \{s\}, E)$ on the agents, where the edge set $E\subseteq (B\cup \{s\})^2$ represents social connections between agents. The {\em neighbour set} of an agent $i\in V$ is $r_i \coloneqq \{j\in B \mid (i,j)\in E\}$. In particular, $r_s$ is the set of all neighbours of $s$. We assume that all buyers $v$ are reachable from $s$ in $G$ via paths.
\end{itemize}

We assume that information regarding the auction is not publicly known and the seller relies on buyers to spread this information to attract potential buyers. Initially, the auction information only reaches buyers in $r_s$. During an auction, the buyers who have the auction information are asked to report their neighbours and valuations. Formally, for buyer $i$:
\begin{itemize}[leftmargin=*]
    \item the {\em true profile} $\theta_i\coloneqq (v_i,r_i)$ is private to the buyer $i$ only;
    \item the {\em reported profile} $\theta'_i\coloneqq (v'_i,r'_i)$ where $v'_i\in \R_+$ and $r'_i \subseteq B$ are the {\em reported} valuation and neighbour set by buyer $i$.
\end{itemize}
The reported profile $\theta'_i$ does not have to be $\theta_i$.  The idea is that the buyer $i$ might try to benefit from the auction by strategically reporting $\theta'_i$. By reporting $r'_i$, buyer $i$ diffuses the auction information to all $j\in r'_i$. Following standard convention \cite{li2017mechanism},  we assume that $r_i'\subseteq r_i$.  Some buyers may not be able to participate in the auction had $i$ misreported her neighbours (i.e., $r'_i \subsetneq r_i$).

Fix the true profiles $\theta\coloneqq (\theta_1,\ldots,\theta_n)$. The {\em global profile} is  the reported profiles of all buyers $\theta'\coloneqq (\theta'_1,\ldots,\theta'_n)$. Given $\theta'$, we build the following directed graph, we call the {\em profile graph} $G_{\theta'}$: The nodes are $V_{\theta'}\coloneqq \{s\}\cup B$; put a directed edge  from $i$ to $j$ in the edge set $E_{\theta'}$ if $j\in r'_i$.
A buyer $i$ is \emph{reachable} if there is a path from the seller $s$ to $i$ in this directed graph.
Only reachable buyers can get the auction information.
Technically, any buyer $j$ that is not reachable should not have a reported profile, but for convenience we assume that they have the {\em silent report} $(v_j', r_j')$ where $v_j'=0$, $r_j'=\varnothing$, indicating the agent $j$'s absence from the auction.

Given a global profile $\theta'$, a mechanism returns payment and allocation rules to buyers in $G_{\theta'}$. Let $\Theta$ denote the set of possible global profiles.
\begin{definition} \label{def:M}
A {\em mechanism} $\M$ consists of two functions $(\pi(\cdot),p(\cdot))$, where the mapping $\pi\colon \Theta\to \{0,1\}^{n}$ is the {\em allocation rule} and $p\colon \Theta\to \R^{n}$ is the {\em payment rule}.
For a global profile $\theta'$, the {\em allocation result} $\pi(\theta')$ is written as $(\pi_1(\theta'),\ldots,\pi_n(\theta'))$ and the {\em payment result} $p(\theta')$ as $(p_1(\theta'),\ldots,p_n(\theta'))$.
\end{definition}
For buyer $i\in B$, when  $\pi_i(\theta')=1$, $i$ {\em wins} an item by paying $p_{i}(\theta')$ and is thus a {\em winner}.  When $\pi_i(\theta')=0$, $i$ gets no item. The value $p_i(\theta')$ can be either positive or negative, denoting either cost or reward of buyer $i$, respectively. When the context is clear, we write $\pi_i$ for $\pi_i(\theta')$ and $p_i$ for $p_i(\theta')$.

An ideal mechanism should meet a number of requirements: First, it should incentivise buyers to participate in the auction, and truthfully report their neighbours and valuations. Then, it should maximally allocate items without causing deficit to the seller. Last, it should achieve a target level of social welfare. To formally define these properties, we introduce the following notions:
\begin{itemize}[leftmargin=*]
\item The {\em utility} $u_i(\theta')$ of the buyer $i$ is defined as $v_{i} \pi_{i}-p_{i}$.
\item The {\em social welfare} $\SW(\theta')$ of the mechanism $\M$ is the sum of the utilities of all the agents, i.e., $\sum_{i=1}^n v_{i} \pi_{i}$.
\item The {\em optimal social welfare} $\SW_{\opt}$ is the sum of the top-$m$ valuations in $\theta$.
\item The {\em revenue} $\RV(\theta')$ is the sum of the payment of all buyers, i.e., $\sum_{i=1}^n  p_{i}$.
\end{itemize}
In the next definition, let $\theta'_{-i}\coloneqq (\theta'_1,\ldots,\theta'_{i-1},\theta'_{i+1},\ldots,\theta'_n)$ denote the profiles of all buyers but $i$.

\begin{definition} \label{def:ic}
Let $\M$ be a mechanism.
\begin{enumerate}[leftmargin=*]
\item $\M$ is {\em incentive compatible} (IC) if for any buyer reporting truthfully is a dominant strategy: for all $i\in B$, all global profiles $\theta'$ and $\theta''$, we have $u_i(\theta_i, \theta_{-i}')\geq u_i(\theta_i'',\theta_{-i}'')$ \footnote{As $r_i''$ may be different from $r_i'$, some agents $j$ who are reachable in $G_{\theta'}$ may become unreachable had we replace $\theta'_i$ with $\theta''_i$. $\theta''_{-i}$ is obtained from $\theta'_{-i}$ by replacing $\theta'_j$ with the silent profile for all such agents $j$.}.
\item $\M$ is {\em individually rational} (IR) if any buyer by reporting truthfully receives  non-negative utility.
\item $\M$ is {\em non-deficit} (ND) if for any global profile $\theta'$, the revenue is non-negative, i.e., $\RV(\theta')\geq 0$.
\item $\M$ is {\em non-wasteful} (NW) if all items are allocated to buyers (up to the number of reachable buyers), i.e., for any global profile $\theta'$, $\sum_{i \in V_{\theta'} \backslash \{s\}}\pi_i(\theta') = \min\{m,|V_{\theta'}|-1\}$.
\item $\M$ is {\em efficient}  if it achieves optimal social welfare, i.e., for any $\theta'$, $\SW(\theta')=\SW_{\opt}$.
\end{enumerate}
\end{definition}
Def.~\ref{def:ic} lays out ideal properties for individuals (e.g., IC and IR) and for the entire network (e.g. ND and NW). For standard auctions (without social network), mechanisms are expected to be efficient. This, however, is impossible for diffusion auctions as no diffusion auction mechanism can simultaneously satisfy IC, IR, ND and efficiency \cite{takanashi2019efficiency}. In subsequent sections, we present a new multi-unit diffusion auction mechanism. 

\section{A New Multi-Unit Diffusion Auction}\label{sec:mudan}

\subsection{Diffusion auction by graph exploration} \label{sec:idea}

\noindent {\bf Earlier mechanisms.} We first define a generic technique for multi-unit diffusion auction. Our design draws lessons from  earlier attempts to the problem.

GIDM \cite{zhao2018selling} and DNA-MU \cite{kawasaki2020strategy} are two prominent mechanisms defined for multi-unit diffusion auction.
Both of these mechanisms fail to ensure the IC property as they potentially grant a buyer the opportunity to manipulate the auction outcome: (1) The mechanisms select winners using a tree structure, called the {\em diffusion critical tree}, which encodes information flow in the graph $G_{\theta'}$.
(2) A misreport by a node down a path (e.g. $f$ in Figure~\ref{fig:single}) may affect the decisions of the mechanism made for nodes that are higher up in the tree (e.g. $a$ in Figure~\ref{fig:single}). (3) This changes the level of competition {\em globally}, which in turn presents an unfair advantage to the untruthful node. Appendix A and Appendix B contain detailed descriptions of these mechanisms along with proofs that they are not IC.
To mitigate the problem above, we need therefore to (a) confine our decisions to buyers within a local area of the social network, and (b) mitigate the competitions within this local area so that a buyer cannot affect other parts of the network. We call this idea {\em competition localisation}.

\vspace{-0.3cm}

\begin{figure}[!htbp]
    \centering
    \includegraphics[width=0.3\textwidth]{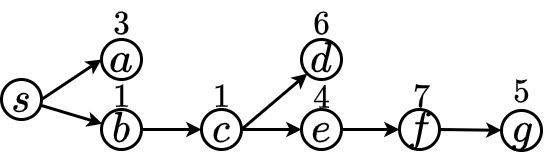}
    \caption{A social network with a seller $s$ and seven buyers (shown in circles). The numbers are the valuations of the circled buyers. GIDM and DNA-MU both fail to ensure truthful reporting here.}
    \label{fig:single}
\end{figure}

\vspace{-0.3cm}

The {\em LDM-Tree} mechanism \cite{liu2022diffusion} applies a form of competition localisation. Specifically, it localise competition within each {\em layer} of the diffusion critical tree, where layer $L_i$ contains agents whose distance from the seller is $i$. The auction runs several rounds. In round $i$, the mechanism only considers nodes in $L_i$ and those nodes in $L_{i+1}$ that do not pose a potential competition to nodes in $L_i$. However, 
LDM-Tree has an obvious flaw:
it severely restricts information diffusion. For instance, if buyers in $L_1$ win all items in the first-layer auction, the outcome of the auction would coincide with the standard VCG mechanism applied {\em only to neighbours of $s$}. This defeats the purpose of diffusion auction.
It will turn out that when applied to single-unit auction (i.e., $m=1$), LDM-Tree may produce outcomes with arbitrarily inferior social welfare than our MUDAN mechanism.
App. C contains detailed descriptions and discussions of LDM-Tree.

\smallskip

\noindent {\bf The generic graph exploration mechanism.} We apply a different form of competition localisation. The idea is to explore the graph $G_{\theta'}$ from the seller $s$, iteratively building a set of {\em explored buyers}. At each iteration, competition is localised to within the explored buyers: 
A {\em winner} is chosen from the explored buyers while some other buyers are {\em exhausted}. The winner and exhausted buyers are incentivised to report their neighbours, which enables more nodes to be explored. Below we explain terms used in a given iteration.

\begin{itemize}[leftmargin=*]
\item  {\bf Explored buyers $A$}: Initially, the set of explored buyers $A=r_s$, i.e., neighbours of $s$. Then at each iteration the set $A$ is updated through  {\em exhausted} and {\em winner} agents (introduced below) using the following procedures:
\vspace{-0.2cm}
\begin{itemize}
    \item Repeatedly adding reported neighbours of exhausted agents until no more buyer can be added.
    \vspace{-0.2cm}
    \item Adding the reported neighbours of the chosen winner.
\end{itemize}

\item {\bf Potential winner set $\PWin$}: At the given iteration, a buyer $i\in A$ is a {\em potential winner} if $i$ is already selected as a winner, i.e., $i\in W$, or may be selected as a winner in the future. The exact definition of $\PWin$ depends on the mechanism and will be made clear in the next sections.

\item {\bf Exhausted agent}: At the given iteration, a buyer in $A\setminus \PWin$ is called an {\em exhausted agents}. The mechanism will ensure that an exhausted agent stays exhausted.


\item {\bf Priority $\sigma_i$}: The algorithm uses {\em priority scores} $\sigma_i$, $i\in \PWin$, to select a winner.
The priority $\sigma_i$ of agent $i$ must satisfy the following:
{\em The value of $\sigma_i$ should be independent of $v_i$ and does not decrease as $|r_i'|$ increases}. A straightforward $\sigma_i$ that meets this condition is $\sigma_i\coloneqq |r'_i|$.\footnote{We introduce other possible $\sigma_i$ and evaluate them empirically in Section~\ref{sec:exp}.}

\item {\bf Winner set $W$}: The buyer with the highest priority score in $\PWin$ is selected as the winner and is added to the set $W$.

\item {\bf Tentative payment $\hat{p}_w$}: When a winner $w$ is selected by the mechanism, a {\em tentative payment} $\hat{p}_w$ is assigned. The tentative payment $\hat{p}_w$ will be used to determine the payment $p_w$ of $w$. The exact definition of $\hat{p}_w$ and $p_w$ will be made clear for each mechanism.


\item {\bf Termination condition}: Terminate the graph exploration if all explored buyers are either winners or exhausted, i.e., $P\setminus W \neq \varnothing$.
\end{itemize}

\begin{algorithm}[htb!]
    \caption{The generic graph exploration mechanism}
    \label{alg:framework}
    \begin{algorithmic}[1]
        \State Initialise $W\leftarrow \varnothing$, $A\leftarrow r_s$ 
        \State {\em Initialise $\PWin$}
        \While{$P\setminus W \neq \varnothing$}\Comment Termination condition \label{line:terminate}
            \While{$A$ contains an unmarked agent $i\in W\cup (A\setminus \PWin)$} \label{line:iteration}
                \State Update $A\leftarrow A\cup r_i'$, mark agent $i$
                \State {\em Update set $\PWin$}
            \EndWhile
            
            \State {\em Assign a priority $\sigma_i$ to each $i \in \PWin$}
            \State Add agent $w\in \PWin\setminus W$ who has the highest priority in $W$
            \State {\em Record tentative payment $\hat{p}_w$}
        \EndWhile
        \State {\em Determine the allocation and payment results using $W$ and tentative payments} 
    \end{algorithmic}
\end{algorithm}

Given the ingredients above, Alg.~\ref{alg:framework} describes the generic graph exploration mechanism. Our MUDAN can be embedded into this framework. Note that the lines in {\em italic} need to be instantiated.

\subsection{The MUDAN mechanism} \label{sec:mudan-mechanism}

We now describe our MUDAN ({\em Multi-Unit Diffusion Auction with No reward}) mechanism for single-demand multi-unit diffusion auction.  MUDAN implements the generic  mechanism (Alg.~\ref{alg:framework}) as follows: Maintain a variable $m'\geq 0$ which records the remaining number of items to be sold. Initially, $m'\coloneqq m$, and is decremented every time a winner is selected. Hence $m-m'$ buyers are selected as winners. At a given iteration, the mechanism sets the following:

\begin{itemize}[leftmargin=*]

\item {\bf Potential winner set $\PWin$}: Rank the buyers in $A\setminus W$ by their valuations such that $v'_{i_1}\geq v'_{i_2} \geq \ldots$. If $|A\setminus W|\leq m'$, then set $\PWin\coloneqq A$; otherwise, add in $\PWin$ the buyers with the top-$m'$ valuations in $A\setminus W$, i.e., $\PWin  \coloneqq W\cup \{i_1,\ldots, i_{m'}\}$.

\item {\bf Tentative payment $\hat{p}_w$}: When a winner $w$ is selected, set $\hat{p}_w$ as the current $(m'+1)$th highest valuation in $A\setminus W$, i.e., $\hat{p}_w\coloneqq v_{m'+1}$.


\end{itemize}
After the last iteration,  we allocate to each winner $w\in W$ an item, i.e., set $\pi_w\coloneqq 1$, and set $p_w\coloneqq \hat{p}_w$. No other buyer will receive an item and the payment to them is 0.  Table~\ref{tab:MUDAN-single} and Figure 5 (in App. D) provide a run-through example of MUDAN. 

\begin{table}
    \begin{center}   
    \caption{Running MUDAN on the network in Fig.~\ref{fig:single} with $m=4$ assuming all buyers report truthfully. Set the priority  $\sigma_i\coloneqq |r_i|$, where buyers with more neighbours get higher priority. `Iter.' shows the iteration number. `Incr. to $A$' column shows the nodes to be added to $A$ in each iteration. The `$\PWin$' column lists potential winners in descending order of $v_i$. The winners are $b,c,e,f$. }
    \label{tab:MUDAN-single}
    \begin{tabular}{|c|c|c|c|c|c|}
    \hline
    Iter. & $m'$ & Incr. to $A$ &  $\PWin$   & $\pi,p$ \\
    \hline
    1 & 4 & $a,b$ & $a,b $ & $\pi_b=1$, $p_b=0$\\
    \hline
    2 & 3 & $c$ & $a,c$ & $\pi_c=1$, $p_c=0$\\
    \hline
    3 & 2 & $d,e$ & $d,e$ & $\pi_e=1$, $p_e=3$\\
    \hline
    4 & 1 & $f$ & $f$   & $\pi_f=1$, $p_f=4$\\
    \hline
    \end{tabular}
    \end{center}
\end{table}


By the definition of $\PWin$, the algorithm terminates when $m'=0$ (after $m$ iterations).
We now show that MUDAN has desirable properties. The next lemma is straightforward (See Appendix D).
\begin{lemma}\label{lem:MUDAN-IR}
   MUDAN satisfies IR, ND, and NW. \QED
\end{lemma}


\begin{lemma}\label{lem:MUDAN-IC}
The MUDAN mechanism satisfies IC.
\end{lemma}
\begin{proof} We prove two statements: {\em 1. A buyer cannot benefit from misreporting her valuation}. Our argument is the following.
Consider an iteration and suppose that $w$, if reporting her profile truthfully, will be selected a winner. We prove that $w$ cannot benefit from misreporting her valuation:
\begin{itemize}[leftmargin=*]
\item If $v_{m'+1}' \leq  v_w'< v_w $ or $v_{m'+1}'\leq v_w < v_w'$,  then $w$ is allocated an item,  pays the $(m'+1)$th highest valuation,  and the utility $u_w((v_w',r_w),\theta_{-w})= u_w((v_w,r_w),\theta_{-w})$.
\item If $w$ reports  $v_w'\leq v_{m'+1}' \leq v_w$, then $w$ loses the item and her utility is $0$.
\end{itemize}
Now consider another buyer  $i\in \PWin \setminus W$, $i\neq w$. We prove that $i$ also cannot benefit from misreporting her valuation:
\begin{itemize}[leftmargin=*]
    \item if $i$ reports valuation $v_i'$ such that $v_{m'+1}' \leq v_i \leq v_i'$ or $v_{m'+1}' \leq v_i'<v_i$, $i$ would still be a potential winner in this iteration, her priority would stay unchanged. 
    \item If $i$ reports valuation $v_i' < v_{m'+1}' < v_i $, she would not be a potential winner and her utility is $0$.
\end{itemize}
Lastly, consider a buyer  $i\in A \setminus \PWin$. We prove that $i$ cannot benefit from misreporting her valuation:
\begin{itemize}[leftmargin=*]
    \item if $i$ reports her valuation $v_i'$ such that $v_i < v_{m'}' \leq v_i'$, then: (i) If she has the highest priority, then her utility is $u_i((v_i',r_i),\theta_{-i})=v_i-v_{m'+1}' <0 =u_i((v_i,r_i),\theta_{-i})$. (ii) Otherwise, her utility remains $0$.
    \item If she reports $v_i<v_i'<v_{m'}'$ or $v_i'<v_i < v_{m'}'$, her utility remains $0$.
\end{itemize}

{\em 2. A buyer cannot benefit from misreporting her neighbours}.
Our argument is the following. Take $i\in A$. If $i$ hides any neighbour, her priority cannot increase.
Consider winner $w$ in a certain iteration. If $w$ hides some of her neighbours and her priority is still the highest, her allocation and payment do not change, so her utility $u_w((v_w,r_w'),\theta_{-w})=u_w((v_w,r_w),\theta_{-w})$. If her priority is not the highest, she loses some items, her utility decreases, i.e., $u_w((v_w,r_w'),\theta_{-w})< u_w((v_w,r_w),\theta_{-w})$.
Now consider agent $i\in A\backslash\{w\}$. If $i$ hides any neighbour, $i$'s priority would not increase and hence, she is still not allocated an item and $u_i((v_i,r_i'),\theta_{-i})=u_i((v_i,r_i),\theta_{-i})=0$.
\end{proof}

\paragraph*{\bf Social welfare.}
We now analyse the social welfare achieved by MUDAN. Note that MUDAN sets the payment $p_i\geq 0$ for any buyer $i\in B$.
This condition means that critical buyers are not incentivised to diffuse the auction information using reward, and thus we call it {\em no-reward} condition. We will show that MUDAN achieves the highest possible social welfare guarantee among IC diffusion auctions with no-reward.
%
Let $w^*\in W$ denote the winner selected in the final iteration. We say that $w^*$ is {\em critical} for a buyer $i$ if all paths from $s$ to $i$ pass through $w^*$.  The next lemma characterises buyers that are explored by the mechanism.

\begin{lemma}\label{lem:local}
A buyer $i$ is explored if and only if $w^*$ is not critical for $i$.
\end{lemma}
\begin{proof}
Suppose $i$ is explored by the mechanism at the $j$th iteration. One can easily prove by induction on $j$ that a path exists from $s$ to $i$ without passing through $w^*$.

Conversely, suppose a path exists from $s$ to $i$ without passing through $w^*$. Let $d_i$ denote the length of the shortest such path. Suppose further $i$ is a node with the smallest $d_i$ that is not explored. Note that $d_i>1$ as all nodes in $r_s$ are explored. Now take the node $j$ that immediately precedes $i$ in the shortest path from $s$ to $i$ without passing $w^*$. Note that $j\in A$ and  $i\in r_j$ by Lemma~\ref{lem:MUDAN-IC}. If $j$ is selected as a winner by the mechanism, then $i$ is explored. Thus $j$ will not be selected as a winner. This will happen only when $v_j\leq v_{w^*}$, which means that $j$ will be exhausted eventually. When $j$ is exhausted, $i$ will be added in $A$. Contradiction.
\end{proof}
Let $B^*$ denote the set of buyers for whom $w^*$ is not critical.


\begin{lemma}\label{lem:weak efficiency}
Suppose a buyer $y\in B$ has a higher valuation than all winners, i.e., $v_y>v_{w}$ for all $w\in W$. Then $y\notin B^*$.
\end{lemma}
\begin{proof}
Take such a buyer $y$ that has the highest valuation. Suppose for a contradiction that a path exists from $s$ to $y$ without passing through $w^*$. By Lemma~\ref{lem:local}, $y$ will eventually be added to $A$. Consider the last iteration before $w^*$ is chosen as the winner. Since $v_y>v_{w^*}$, $w^*$ would not be an element of $\PWin$. Contradiction.
\end{proof}

Lemma~\ref{lem:weak efficiency} describes how MUDAN may fail to achieve optimal social welfare: There exists a buyer $i\in B\setminus B^*$ who has a high (top-$m$) valuation. This motivates us to define the following weakened notion of efficiency.

\begin{definition}\label{def:WAE}
Let $\SW_{\wopt}$ denote the sum of the top-$m$ valuations among buyers in $B^*$.
 A mechanism $\M$ is {\em $\epsilon$-weakly efficient} if for any global profile $\theta'$, we have $\SW(\theta') \geq \epsilon \SW_{\wopt}$.
\end{definition}
Weak efficiency means achieving the highest social welfare among the explored buyers.
By Lemma~\ref{lem:weak efficiency}, MUDAN selects the buyer with the highest valuation from $A$ as a winner, thus achieving $1/m$-weak efficiency. 
The only currently-known IC multi-unit diffusion auction mechanism, LDM-Tree, may produce outcomes for the single-unit case that are arbitrarily inferior than  MUDAN in terms of social welfare. See Appendix C. 

The theorem below summarises the results above.

\begin{theorem}\label{thm:MUDAN}
MUDAN terminates within time $O(n^2+|E|)$, satisfies IC, IR, ND, NW, and $1/m$-weak efficiency. \QED
\end{theorem}

Lastly, we show that $1/m$-weak efficiency is as good as it can be for IC and ND diffusion auctions with no-reward. 
\begin{theorem}\label{thm:1/m}
For any $m\geq 1$ and any constant $\lambda>0$, there exists profile $\theta$ where no $m$-unit IC and NW diffusion auction with no-reward achieves $(1/m+\lambda)$-weak efficiency. 
\end{theorem}

\begin{figure}[!htbp]
    \centering
    \includegraphics[width=0.5\columnwidth]{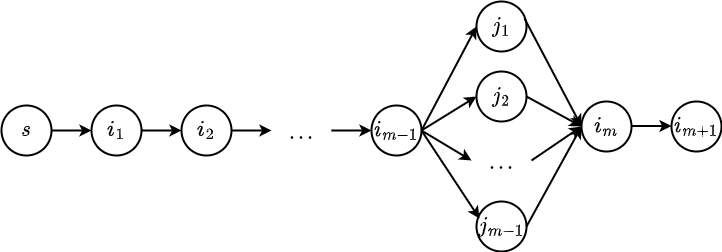}
    \caption{A social network with a seller $s$ and $2m$ buyers.
    }
    \label{fig:subopt}
\end{figure}

\begin{proof}
Consider the graph as shown in Fig.~\ref{fig:subopt} which depicts a situation with $2m$ buyers. The valuations and connections of the buyers are defined as follows:
\begin{itemize}[leftmargin=*]
    \item $i_1,i_2,i_3,\ldots,i_{m-1}$ have valuation $n>1$,
    \item $j_1,j_2,\ldots,j_{m-1}$ have valuation $n^2-\tau$ for a small $\tau>0$,
    \item $i_m$ has valuation $n^2$, and $i_{m+1}$ has valuation $n^3$.
\end{itemize}
To guarantee IC and NW, a mechanism must allocate $m-1$ items to buyers $i_1,i_2,\ldots,i_{m-1}$ and the last item to $i_m$.
Thus the social welfare of any IC mechanism is at most $n^2+(m-1)n$. For such a mechanism, $\SW_{\wopt}$ is $m n^2 - (m-1)\tau$.
For sufficiently large $n$,
\begin{align*} 
      & \frac{n^2+(m-1)n}{m n^2-(m-1)\tau} 
    = \frac{1}{m} + \frac{(m-1)\tau/m+(m-1)n}{m n^2-(m-1)\tau} \\
                                       & \leq \frac{1}{m}+\frac{1}{mn^2-m}+\frac{1}{n-(m-1)\tau/mn} 
                                       < \frac{1}{m}+\lambda.
\end{align*}
\end{proof}

{\bf Remark.} A natural question arises as to whether a mechanism exists which achieves efficiency. As demonstrated above, such a mechanism necessarily incentivises critical buyers using rewards. We provide a mechanism, named MUDAR, towards meeting this goal. Similar to MUDAN, MUDAR also implements the generic mechanism (Alg.~\ref{alg:framework}). The difference between the two mechanism lies in how they incentivise winners. Unlike MUDAN, once a winner $w$ is chosen,  MUDAR may either allocate an item to $w$ or give $w$ a reward (i.e., a negative payment), which equals to the utility of $w$ had she been allocated an item. The allocation result is determined after the graph exploration is completed, when the buyers' connections are fully revealed.  In this way, MUDAR can identify  buyers that report the $m$-highest valuations globally.
MUDAR achieves IR, ND, and NW. Assuming that the $m$th highest valuation among buyers is public, MUDAR ensures a weaker form of IC. Moreover, the buyers form an equilibrium such that applying MUDAR leads to an efficient allocation. See details in Appendix E.

\section{Multi-demand multi-unit diffusion auction}
\label{sec:multi}

We now generalise the {\em single-demand} setting in Section~\ref{sec:model} to {\em multi-demand} auction.
Here each buyer $i\in B$ may demand more than one item and attaches a valuation $v_{i,j}$ to the $j$th item she gets, where $1\leq j \leq m$. The valuation for all items is denoted by a vector $\vec{v}_i \coloneqq (v_{i,1},\ldots,v_{i,m})\in \R_{+}^m$; we call $\vec{v}_i$ the {\em valuation vector} of the buyer $i$.
The valuation vector of a buyer who demands $k < m$ units is represented by an $m$-dim vector with $m-k$ $0$s at the end.
For simplicity, we omit $0$ in the vector.
Each additional unit often brings less additional utility than that from the previous unit, which is known as the {\em law of diminishing marginal utility} in micro-economics  \cite{gossen1983laws}. \footnote{As an analogy, think of a hot summer, the first bottle of iced water brings much satisfaction while the second or even third one brings lower satisfaction.}
Therefore, we assume that the buyers have {\em diminishing valuation} towards the items: $v_{i,j} \geq v_{i,j+1}$ for $j=1,\ldots, m-1$.
In this setting, we denote the {\em (multi-demand) profile} of a buyer $i$ as $\eta_i\coloneqq (\vec{v}_i, t_i)$ where $\vec{v}_i$ is the valuation vector and $t_i\subseteq B$ is the set of neighbours of $i$. The global profile is $\eta'\coloneqq (\eta'_1,\ldots, \eta'_n)$ that corresponds to profile graph $G_{\eta'}$.

Let $H$ denote the set of all possible (multi-demand) global profiles.
A {\em mechanism} $\M$ in this setting consists of {\em allocation rule} $\pi\colon H\to \{0,1\}^{n\times m}$ and {\em payment rule} $p\colon H\to \R^{n\times m}$. Here, each $\pi_i (\eta')$, $p_i(\eta')$ are $m$-dimensional vectors; We write them as $\pi_i (\eta')=(\pi_{i,1} (\eta'), \ldots, \pi_{i,m} (\eta'))$ and $p_i (\eta')=(p_{i,1} (\eta'), \ldots, p_{i,m} (\eta'))$.

We formally define properties of a mechanism below:
\begin{itemize}[leftmargin=*]
\item The {\em utility} $u_i(\eta')$ of the buyer $i$ is defined as $\sum_{j=1}^m v_{i,j} \pi_{i,j}-p_{i,j}$.
\item The {\em social welfare} $\SW(\eta')$ of the mechanism $\M$ is the sum of the utilities of all the agents, i.e., $\sum_{i=1}^n \sum_{j=1}^m v_{i,m} \pi_{i,j}$.
\item The {\em optimal social welfare} $\SW_{\opt}$ is the sum of the top-$m$ valuations among $v_{i,j}$, $1\leq i\leq n$, $1\leq j\leq m$.

\item The {\em revenue} $\RV(\eta')$ is the sum of the payment of all buyers, i.e., $\sum_{i=1}^n\sum_{j=1}^m  p_{i,j}$.
\end{itemize}

Instead of designing a mechanism for the multi-demand setting from scratch, we reduce the problem to its single-demand counterpart. Given (multi-demand) profiles $\eta$ of buyers $B=\{1,\ldots,n\}$, our goal is to construct a set of buyers $\widetilde{B}$ with (single-demand) profiles $\theta$ in such a way that
any mechanism $\widetilde{\M}$ that is applied to $\widetilde{B}$ corresponds to a mechanism $\M$ that is applied to $B$. The intuition is that, essentially, we may view a buyer $i\in B$ as $m$ buyers, each demanding one item with valuation $v_{i,j}$.
%
More precisely, given a profile $\eta'$, to define our mechanism $\M$ we perform the following steps.

\smallskip

\noindent {\bf (1).}
For each buyer $i\in B$, create $m$ nodes $i_1,\ldots,i_m$ in $\widetilde{B}$, each corresponding to an item $1\leq j\leq m$, i.e., $\widetilde{B}\coloneqq \{i_j\mid 1\leq i\leq n, 1\leq j\leq m\}$.

\noindent {\bf (2).} Construct the following profiles $\theta'$ for buyers in $\widetilde{B}$:
\begin{itemize}
\item In the profile graph $G_{\theta'}$, connect all buyers $i_1,\ldots,i_m$ to form a chain using edges $(i_1,i_2),\ldots, (i_{m-1},i_m)$.
\item If an edge exists between the seller $s$ and buyer $i$ in $G_{\eta'}$, then add an edge $(s,i_1)$ in $G_{\theta'}$.
\item If an edge exists between buyers $i$ and $j$ in $G_{\eta'}$, then add an edge $(i_m,j_1)$ in $G_{\theta'}$.
\item The true and reported valuation of a buyer $i_j\in \widetilde{B}$ are $v_{i,j}$ and $v'_{i,j}$, resp.
\item The priority of the buyer $i_j$, for $1\leq j\leq m$, is  the priority of $i$ in $B$.
\end{itemize}

\noindent {\bf (3).} Apply a (single-demand) mechanism $\widetilde{\M}=(\tilde{\pi},\tilde{p})$ to $\theta'$, where $\tilde{\pi}$ is the allocation rule and $\tilde{p}$ is the payment rule. Return the mechanism $\M\coloneqq (\pi, p)$ where allocation rule $\pi$ and payment rule $p$ are defined below:
\begin{itemize}
\item Define $\pi\colon H\to \{0,1\}^{n\times m}$ by $\pi_{i,j}(\eta')\coloneqq \tilde{\pi}_{i_j}(\theta')$.

\item Define $p\colon H\to \R^{n\times m}$ by $p_{i,j}(\eta')\coloneqq \tilde{p}_{i_j}(\theta')$.
\end{itemize}
Figure~\ref{fig:multi} illustrates the construction of $\theta'$ given a multi-demand profile $\eta'$ over seven buyers. When we choose  MUDAN  as $\widetilde{\M}$, the corresponding mechanism $\M$ is called {\em MUDAN-$m$}.


\begin{figure}[!htbp]
    \centering
    \includegraphics[width=0.45\columnwidth]{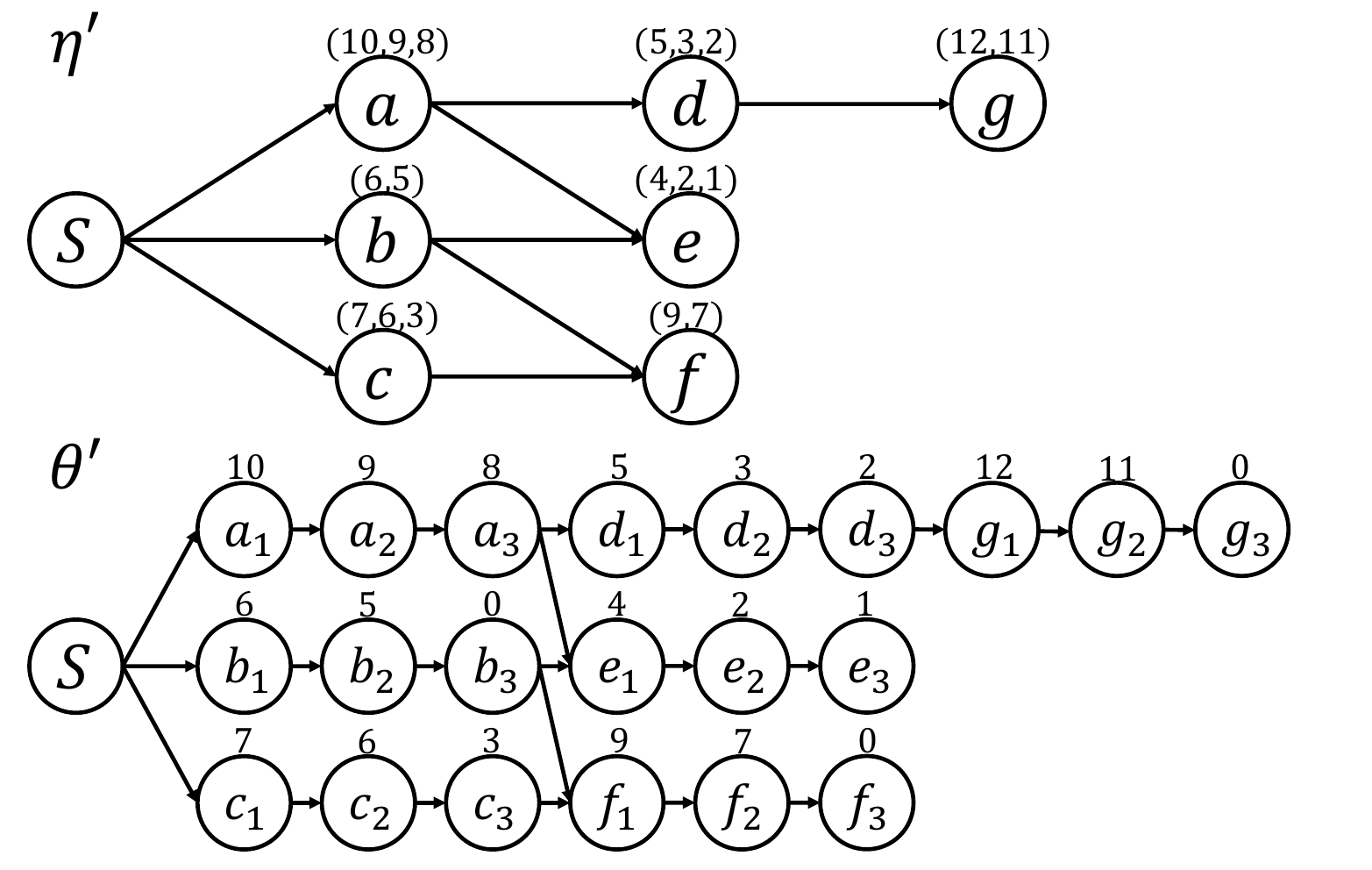}
    \caption{The reduction of an multi-demand auction (above) to the single-demand setting (below). Valuation is on top of each node.}
    \label{fig:multi}
\end{figure}

Consider a mechanism $\widetilde{\M}$ applied to $\widetilde{B}$ with the constructed profile $\theta$, and the corresponding mechanism $\M$ applied to $B$ with the profile $\eta$. It is straightforward to observe from the construction that the utility of a buyer $i$ from $\M$ equals the sum of utilities of $i_1,\ldots,i_m$ from $\widetilde{\M}$. Additionally, the social welfare, optimal social welfare, revenue, and number of allocated items of $\M$ are equal to those of $\widetilde{\M}$, respectively. Please refer to Lemma 18 in Appendix F for the formal statement.
%
Moreover, analogously to Def.~\ref{def:ic}, we define IC, IR, ND, NW and efficiency for a mechanism in multi-demand case. See Def. 19 of App. F. The next lemma follows from Lem. 18. 

\begin{lemma}\label{lem:multi-property}
$\M$ satisfies any of the properties in IR, ND, NW, efficiency whenever so does $\widetilde{\M}$. \QED 
\end{lemma}

\paragraph*{\bf The MUDAN-$m$ mechanism.}
We show the MUDAN-$m$ mechanism has desirable properties (in Theorem~\ref{thm:mudan-m}). 
IR, ND, NW, and $1/m$-weak efficiency easily carry over from the single-demand case. The proof of IC, however, requires a non-trivial justification: If a buyer $i\in B$ misreports her valuation vector $\vec{v}_i$, a group of buyers in $\widetilde{B}$, namely $i_1,\ldots,i_m$, may misreport their valuations {\em together}. This amounts to a case of {\em collusion} among the buyers in $\widetilde{B}$, which is not accounted for in the single-demand case. Nevertheless, in Lemma~\ref{lem:mudan-m-ic} below, we show that the buyers in $B$ would not violate the truthfulness properties. In particular, our mechanism ensures that at most one buyer $i_j$ where $1\leq j\leq m$ may be in the set $\PWin\setminus W$ for any $i\in B$ at any iteration, hence collusion does not give extra incentive for the buyers $i_1,\ldots,i_m$.  See the full proof in Appendix F.

\begin{lemma}\label{lem:mudan-m-ic}
The MUDAN-$m$ mechanism is IC. 
\end{lemma}

\begin{proof}
We first prove that no buyer in $B$ can benefit from misreporting her valuation vector.
To run $\M$, we first execute the MUDAN mechanism $\widetilde{\M}$ given the constructed profile $\theta'$ on $\widetilde{B}$. Now consider the execution of $\widetilde{\M}$.

If $i_1$ is selected as a winner, then $\widetilde{\M}$ will proceed to select more $i_j$ as winners until the valuation is so low that $i_j$ is not added in $\PWin$. That is because after $i_j$ is chosen as winner, the only buyer in $\widetilde{B}$ that is added to $A$ is $i_{j+1}$. By definition of the priority on $B$, if $i_{j+1}$ is added to $\PWin$, then $i_{j+1}$ will be the next node with the highest priority and is thus chosen  as the winner.

Consider a given iteration of $\widetilde{\M}$. By the argument above, if a buyer $i_j\in \PWin$ is not chosen as the winner, then $j=1$. Moreover, if $i_j$ is exhausted, then $i_{\ell}$ for all $\ell>j$ are exhausted also due to the diminishing valuation. This means that for any buyer $i\in B$, either $i_1\in \PWin$, or $\{i_1,\ldots,i_m\}\cap \PWin=\varnothing$. This observation is crucial for the proof of truthfulness.

Now suppose in the given iteration, $w_j$ is chosen as the winner. Suppose $w\in B$ misreports her valuation vector so that $\vec{v'_w}\neq \vec{v_w}$. We show that this strategy will not give $w$ extra utility at this iteration:
{\bf (a)} Suppose $v_{w,\ell}\neq v'_{w,\ell}$ where $\ell>j$. At this iteration, $w_\ell$ would {\em not} have been added in $A$. Therefore for $w$, misreporting $v_{w,\ell}$ does not affect the tentative payment $\hat{p}(w_j)$ for her $j$th item. 
{\bf (b)} Suppose $v_{w,k}\neq v'_{w,k}$ where $k<j$. Note that $w_k$ must have been a winner in an earlier iteration. So the valuation of $w_k$ is not taken into consideration when the algorithm sets the tentative payment $\hat{p}(w_j)$. 
{\bf (c)} Suppose $w$ misreports the valuation of $v_{w,j}$ in such a way that $v'_{w,j}>v_{w,j}>\hat{p_{w_j}}$ or $v_{w,j}>v'_{w,j}>\hat{p_{w_j}}$. Here, $\hat{p_{w_j}}$ is the tentative payment of $w_j$ assuming $\tilde{\M}$ sets $w_j$ as the winner. As this does not change $w_j$'s priority, $w$ is still chosen as the winner in this iteration and is allocated her $j$th item with payment $\hat{p_{w_j}}$. This item will add $w$ a utility of $u_{w_j}((v'_{w_j},r_{w_j}), \theta_{-w_j})=u_{w_j}((v_{w_j},r_{w_j}), \theta_{-w_j})$.
{\bf (d)} On the other hand, if  $v'_{w,j}<\hat{p_{w_j}}<v_{w,j}$, then $w$ loses the $j$th item as all $w_{\ell}$ where $\ell\geq j$ are exhausted, giving her no extra utility.
Summarising (a)-(d), in the given iteration, misreporting any element in valuation vector $(v_{w,1},\ldots,v_{w,m})$ will not give $w$ extra utility.

We then prove that for buyer $i\neq w, i_1\in \PWin$ and for $i\neq w, \{i_1,\ldots,i_m\}\cap \PWin=\varnothing$, $i$ cannot benefit from misreporting valuation using similar arguments above. See full argument in App. F.


It remains to prove that no buyer in $B$ can benefit from misreporting her neighbour set. Our argument is the following: For each agent $i_j\in \widetilde{B}$, her priority cannot increase when she hide any of her neighbours. And her neighbours is only added in $A$ either when $i_m$ is chosen as a winner or when $i_m$ is exhausted so that her neighbours cannot influence her allocation. The rest of the proof is the same as in the proof of Lemma~\ref{lem:MUDAN-IC}.
\end{proof}


To analyze the social welfare of MUDAN-$m$, we introduce the definition of $\epsilon$-weak efficiency. Similarly to the single-demand case, we define critical buyers, the set $B^*$, weakly-optimal social welfare $\SW_{\wopt}$, and $\epsilon$-weak efficiency for the multi-demand case. The only difference is that $\SW_{\wopt}$ is the sum of the top-$m$ valuations among $v_{i,j}$ where $i\in B^*, 1\leq j\leq m$. See formal statement in Appendix F.
%
%
The next theorems is also proved in Appendix F.

\begin{theorem}
\label{thm:mudan-m}
MUDAN-$m$ is IC, IR, ND, NW, and $1/m$-weakly efficient.\QED
\end{theorem}


\section{Experiments}\label{sec:exp}

\begin{figure}[!htbp]
    \centering
    \includegraphics[width=1\textwidth]{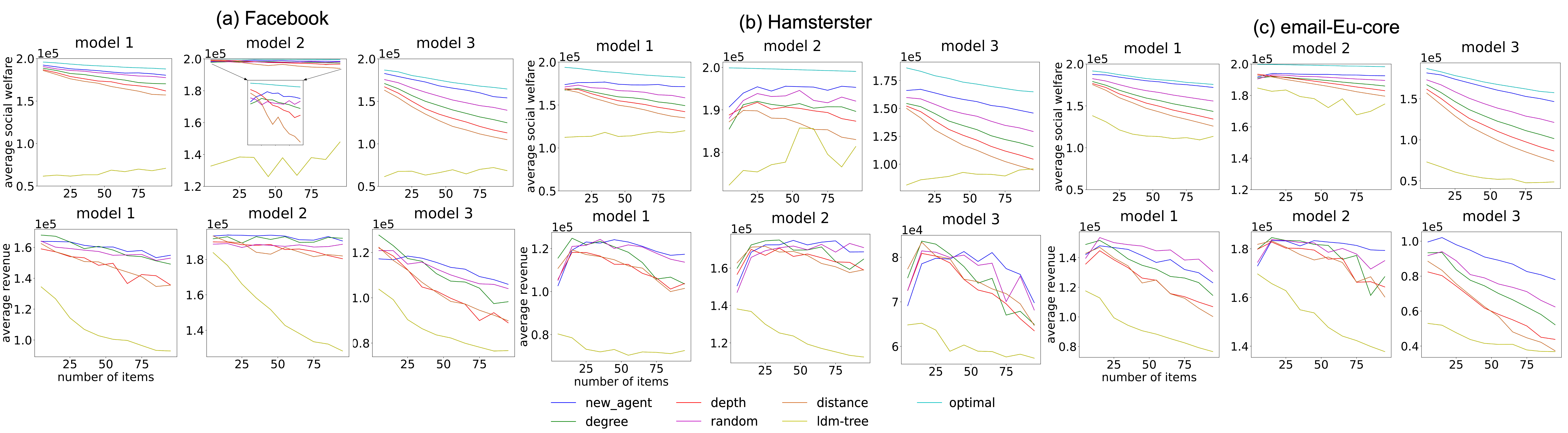}
    \caption{Social welfare and revenue of different traversal strategies in three models for Facebook, Hamsterster and Email networks}
    \label{fig:results}
\end{figure}

Finally, we empirically evaluate our mechanism\footnote{The code of this work is available at \url{https://anonymous.4open.science/r/MUDAN-2C0B}.}. We have two purposes: First we evaluate the social welfare and revenue when MUDAN-$m$ is applied, and second we focus on priority score $\sigma_i$. Recall that the priority $\sigma_i$ determines the selection of winner at an iteration. It affects the outcome of MUDAN-$m$. 
%
The idea of priorities to buyers has been exploited by several existing diffusion auction mechanisms, and three priority orderings were used:  {\bf(1)} {\bf depth-based selection} \cite{zhao2018selling}
{\bf (2)} {\bf distance-based selection} \cite{kawasaki2020strategy}, and {\bf (3)} {\bf degree-based selection} \cite{xiao2022multi}.
Yet their effectiveness has not been analysed. Since all three approaches can be adopted by our mechanisms, we now examine how they affect social welfare in MUDAN-$m$.

Depth-based (or distance-based) selection prioritizes buyers who are farthest away from (or closest to) the seller, which poses a risk of omitting buyers with high valuations but who are close to (or far away from) the seller. Degree-based selection prioritizes buyers who have more neighbors, but this does not necessarily guarantee the exploration of a large number of buyers. With this in mind, we propose a new traversal strategy: {\bf (4)} {\bf New-agent-based selection:} Prioritize agents who can bring the highest number of unexplored buyers to $A$. In this strategy, only the contribution to graph exploration is taken into account, which we expect to lead to higher social welfare. We describe our experiments below.


{\bf Dataset.} We use three real-world datasets, including
Facebook social network \cite{leskovec2012learning}, Hamsterster friendships \cite{Kunegis2013konect}, and email-Eu-core network \cite{Yin2017Local}. Facebook network has $4,039$ nodes and $88,234$ edges, Hamsterster friendship network has $1,858$ nodes and $12,534$ edges, while email-Eu-core network has $1,005$ nodes and $25,571$ edges. Table \ref{tab:dataset} shows the key statistics of these three datasets.
For each dataset, we randomly select one node as the seller. As the initial setup, especially the neighbour set of the seller, may effect experiment results, we repeat each scenario $n/2$ times, where $n$ is the number of nodes, and calculate the average revenue and social welfare as the result for the scenario.

\begin{table}
    \centering
    \begin{tabular}{|c|c|c|c|c|}
    \hline
    dataset & $|V|$ & $|E|$ & $C$ & diameter \\
    \hline
    Facebook social network & 4039 & 88234 & 0.6055 & 8 \\
    \hline
    Hamsterster friendships & 1858 & 12534 & 0.0904 & 14 \\
    \hline
    email-Eu-core network & 1005 & 25571 & 0.3994 & 7  \\
    \hline
    \end{tabular}
    \caption{Dataset statistics. $C$ denotes clustering coefficient}
    \label{tab:dataset}
\end{table}


{\bf Valuation.}
We use three different models to generate the agents' valuation.
\noindent {\bf Model 1:} All the valuations of buyers are sampled at random i.i.d. Specifically, we assume $v_{i,j} \sim U(0,200000)$ for $1\leq i \leq n, 1\leq j\leq m$.
\noindent {\bf Model 2:} To increase diversity, the highest valuations of buyers are sampled i.i.d. while their subsequent valuations are independently but non-identically distributed. Here, we assume that $v_{i,1} \sim U(0,200000)$ for $1\leq i \leq n$, and $v_{i,j} \sim U(1, v_{i,1})$ for $1\leq i \leq n, 1< j \leq m$.
\noindent {\bf Model 3:} The valuation of the buyers are affected by their neighbours, in particular, the {\em homophily principle} asserts that agents who are tightly connected tend to exhibit similar preferences \cite{mcpherson2001birds}. To capture possible dependence among closely-tied buyers, we deploy the DeGroot model, an established model of social influence \cite{jackson2010social}, to generate the highest valuations $v_{i,1}$ for $1\leq i \leq n$.
 DeGroot model \cite{jackson2010social} assumes that each agent's valuation for the next iteration is derived from a weighted average of her own valuation and those of her neighbours in the network. The weight is assigned by the agent, and it represents her confidence in her own valuation or her friendship with others.
The other valuations of a certain buyer $v_{i,j} \sim U(1, v_{i,1})$ for $1\leq i \leq n, 1< j \leq m$.

{\bf Benchmark.} The only currently known IC multi-unit diffusion auction mechanism, LDM-Tree \cite{liu2022diffusion}, is chosen as our benchmark.
{\em Random selection} is used as a benchmark for traversal strategies. It randomly allocates a buyer from the potential winner set.
The {\em optimal social welfare}, i.e., the sum of top-$m$ valuations, is also included in the comparison.
We evaluate all four implementations of MUDAN-$m$ as well as these benchmarks in terms of {\em social welfare $\SW$} and {\em revenue $\RV$} as defined in Sec.~\ref{sec:model}.





{\bf Results.} Figure~\ref{fig:results} shows the average $\SW$ and $\RV$ per item. {\bf (1)} As shown, MUDAN-$m$ significantly outperforms LDM-Tree. 
MUDAN-$m$ with new-agent selection loses by at most $9\%$, $12\%$ and $8\%$ from the optimal social welfare in Facebook, Hamsterster, and Email networks, respectively, exhibiting that it achieves near-optimal social welfare. In contrast, LDM-Tree loses nearly $75\%$, $62\%$ and $62\%$, resp. 
 {\bf (2)} Among the priority schemes, new-agent selection in general outperforms the other  schemes. When the number of items is small, new-agent selection is slightly less than that of other strategies, but grows faster as the number of items increases. 
{\bf (3)} Across the valuation models, MUDAN-$m$ with new-agent selection performs better in general. This is particularly visible in model 3, which is consistent with our expectation: When the valuations are dependant, buyers with similar valuations form communities, new-agent-based selection is more advantageous as it could more easily jump out from lower-valuation communities.
We may draw consistent conclusions from the results for all three datasets showing robustness of our mechanism.

%
%

\section{Conclusion and future work}
We focus on multi-unit diffusion auctions in this paper. We propose the MUDAN mechanism, the first multi-unit diffusion auction that satisfies the properties IC, IR, ND, NW, and $1/m$-weak efficiency; We gave a case when the $1/m$-weak efficiency bound is tight. We also define a reduction from multi-demand case to single-demand case so that MUDAN can be employed to this generalised problem. The corresponding MUDAN-$m$ mechanism satisfies all properties in the single-demand case. 



As future work, we plan to explore the cases (1) when buyers can perform {\em false name attacks}, i.e., reporting agents who are not in their neighbour set,  (2) when buyers may collude to benefit from group misreporting, and (3) when the values of the items are not additive.


\section*{Acknowledgements}
This work is partially supported by National Natural Science Foundation of China No. 62172077.

\bibliographystyle{abbrv}
\bibliography{sample.bib}

\appendix

\clearpage
\newpage

\noindent {\LARGE \bf Appendix}

\section{GIDM Mechanism}
\label{app:gidm}

For any global profile $\theta'$ and its
profile digraph $G_{\theta'}=(V_{\theta'}, E_{\theta'})$, we define the {\em $\theta'$-criticality} relation on the buyer set $V_{\theta'}$:
\begin{definition}
For any buyers $i,j\in V_{\theta'}$, $i$ is {\em $\theta'$-critical to} $j$, denoted by $i\preceq_{\theta'} j$, if all paths from $s$ to $j$ in $G_{\theta'}$ go through $i$.
\end{definition}
The criticality relation can be visualised as a tree structure $T_{\theta'}$ which we call {\em $\theta'$-diffusion critical tree} \cite{zhao2018selling}. The root of $T_{\theta'}$ is $s$ and for any buyer $j\in V_{\theta'}$, the parent of $j$ is the node $i\preceq_{\theta'} j$ that has the closest distance to $j$.
$T_{\theta'}$ provides insights on how competitions can be alleviated among the buyers to elicit truthful behaviours: If $i\preceq_{\theta'} j$, then $i$ has the power to block information diffusion to $j$. A truthful mechanism should thus make $i$ indifferent on whether $j$ participates the auction or not. Intuitively,
if $i$'s valuation is sufficiently high so that $i$ may win an item had $i$ misreported her connections, then $i$ should be compensated, in the forms of rewards or allocation.

GIDM allocates items based on the diffusion critical tree $T_{\theta'}$. Essentially, the mechanism treats each subtree of $T_{\theta'}$ as a {\em sub-market} and allocates items recursively for each sub-market. The number of items allocated to each sub-market depends on the number of top-$m$ bidders that subtree contains. An agent $i$ becomes a winner either if she is one of the top-$m$ bidders, or her reported valuation is high enough to win had she misreported her neighbours. A buyer who does not win any item may also receives a reward. See detailed procedures below.

\smallskip

\noindent {\bf Detailed procedure description:} Given a global profile $\theta'$, the mechanism first constructs the tree $T_\theta'$, which includes only the buyers with top $m$ valuations and their critical parents. For each buyer $i\in T_\theta'$, let $w_i$ denote the total number of items that buyer $i$ and her critical children $\mathsf{Desc}_i$ can get in the optimal allocation tree. Then the allocation is done with a DFS-like procedure. Let $S$ be a stack that the mechanism used for DFS-like traverse. Initially, $S$ is empty. The seller $s$ first gives $w_i$ items to buyers $\{i\in r_s \mid w_i>0\}$ and pushes all of the buyers into $S$. Then the mechanism pops the top buyer $i$ from $S$ and checks whether she can be allocated an item in the optimal allocation. If so, the mechanism allocates her an item and adds her into the winner set $W$; otherwise, check whether $i$ can receive an item if removing a subset $\mathsf{Desc}_i^m$ of her critical children: $\mathsf{Desc}_i^m \coloneqq \mathsf{Desc}_i(\theta)^m\cup \mathsf{Ance}(\mathsf{Desc}_i(\theta)^m)\cup \mathsf{Desc}(\mathsf{Ance}(\mathsf{Desc}_i(\theta)^m)),$ where $\mathsf{Desc}_i(\theta)^m$ contains the top-$m$ critical children of $i$.  If so, she gets an item from her critical children who has the lowest valuation and the mechanism updates the other agents' weight. If not, she passes the item to her critical children $C_i$ according to their weight and adds them into $S$. At the same time, she gets some reward. Repeat the allocation process until $S$ is empty. See Algorithm~\ref{alg:GIDM}.

\begin{algorithm}[ht!]
    \caption{GIDM mechanism}
    \label{alg:GIDM}
    \begin{algorithmic}[1]
    \State Initialise empty stack $S$
    \State Construct optimal allocation tree $T_\theta'$ and calculate each agent $i$'s weight $w_i$
    \State Push buyers from $\{i\in r_s \mid  w_i>0\}$ into $S$
    \While{$S\neq \varnothing$}
    \State Pop the top buyer $i\in S$ from $S$
    \If{$i$ can get an item in optimal allocation}
    \State Allocate $i$ an item with payment $\SW_{-D_i}-(\SW_{-\mathsf{Desc}_i^m}-v_i)$, add her into winner set $W$
    \Else
    \If{$i$ can get an item if removing $\mathsf{Desc}_i^m$}
    \State Allocate $i$ an item from the agent who who has the lowest valuation with payment
    $\SW_{-D_i}-(\SW_{-\mathsf{Desc}_i^m}-v_i)$
    \State Update $w_i$ of relative agents
    \Else
    \State Pass $i$' item(s) to her critical children, reward $i$ by $\SW_{-D_i}-\SW_{-\mathsf{Desc}_i^m}$, where $D_i=\mathsf{Desc}_i\cup\{i\}$
    \EndIf
    \EndIf
    \State Push $i$'s critical children into $S$
    \EndWhile
    \end{algorithmic}
\end{algorithm}

This mechanism is not IC as it enables a buyer with a high valuation to potentially manipulate the outcome to her own benefit \cite{takanashi2019efficiency,kawasaki2020strategy}. We show it with Example~\ref{exa:GIDM} below.

\begin{example} \label{exa:GIDM}
Consider the social network in Fig.~\ref{fig:single} with seven buyers. As the network is a tree, $G$ coincides with the diffusion critical tree $T_\theta$. The seller has $m=4$ items to sell.  Suppose we run GIDM on this social network. If all buyers report truthfully, GIDM gives
all four items to the sub-market of $b$ as all top-four bidders (i.e., the vip) $d,e,f,g$ are descendents of $b$. The resulting winners are $b,c,d,e$.
On the other hand, suppose $f$ misreport her connection declaring $r'_f=\varnothing$, then GIDM will only give three items to the sub-market of $b$, while allocating an item to $a$. This increases the competitions within the sub-market of $b$ and the allocation result is $a,b,c,f$. In this way, $f$ manipulates the market to her advantage. The detailed run-through of these two cases are shown in Table~\ref{tab:GIDM} and Table~\ref{tab:GIDM-2}.
\end{example}

\begin{table}
    \begin{center}
    \caption{Running GIDM on the network in Fig.~\ref{fig:single} with $m=4$ assuming all buyers report truthfully. The winners are $b,c,d,e$.}
    \label{tab:GIDM}
    \begin{tabular}{|c|c|c|c|}
    \hline
    rd & $m'$ &  winners & $\pi,p$ \\
    \hline
    1 & 4 & $\varnothing$ & $\pi_b=1$, $p_b=0$ \\
    \hline
    2 & 3 & $\{b\}$ & $\pi_c=1$, $p_c=0$\\
    \hline
    3 & 2 & $\{b,c\}$ & $\pi_d=1$, $p_d=5$\\
    \hline
    4 & 1 & $\{b,c,d\}$ & $\pi_e=1$, $p_e=3$ \\
    \hline
    \end{tabular}
    \end{center}
\end{table}

\begin{table}
    \begin{center}
    \caption{Running GIDM on the network in Fig.~\ref{fig:single} with $m=4$ assuming $f$ reports $r_f'=\varnothing$. The winners are $a,b,c,f$.}
    \label{tab:GIDM-2}
    \begin{tabular}{|c|c|c|c|}
    \hline
    rd & $m'$ & winners & $\pi,p$ \\
    \hline
    1 & 4 & $\varnothing$ & $\pi_a=1$, $p_a=0$ \\
    \hline
    2 & 3 & $\{a\}$ & $\pi_b=1$, $p_b=0$\\
    \hline
    3 & 2  & $\{a,b\}$ & $\pi_c=1$, $p_c=0$\\
    \hline
    4 & 1 & $\{a,b,c\}$ & $\pi_f=1$, $p_e=6$ \\
    \hline
    \end{tabular}
    \end{center}
\end{table}

\section{The DNA-MU Mechanism}
\label{app:dna-mu}
DNA-MU follows a conceptually simpler process. The main difference lies in assigning a priority order to buyers based on their distance from the seller in the social network. Buyers who are closer to the seller get higher priority. In this priority order, the mechanism traverses through all agents while allocating winners. 

\smallskip

\noindent {\bf Detailed procedure description:} 
DNA-MU first initialises the residual supply $m'=m$.
For each agent $i$, DNA-MU extracts the $m'$th highest valuation, called $p_i$, among buyers in $T_{\theta'}$ who fall outside of the subtree of $i$. If $i$'s declared valuation $v'_i\geq p_i$, then $i$ has an incentive to misreport her neighbours as doing so will guarantee her winning an item. In this case, the mechanism compensates $i$ by allocating an item to $i$ with payment $p_i$, before reducing $m'$ by 1. Otherwise, the mechanism proceeds to the next node in the priority order until $m'=0$. See Algorithm~\ref{alg:DNA-MU}.

\begin{algorithm}[ht!]
    \caption{DNA-MU mechanism}
    \label{alg:DNA-MU}
    \begin{algorithmic}[1]
    \State Initialise residual supply $m'\leftarrow m$, winner set $W$
    \For{agent $i\in N$ according to the priority}
    \If{$m'=0$}
    \State Break
    \EndIf
    \State $p_i$ is the $m'$th highest valuation in $N_{-i}\backslash W$
    \If{$v_i>p_i$}
    \State Allocate $i$ an item with payment $p_i$
    \State Update $W\leftarrow W\cup {i}$ and $m'\leftarrow m'-1$
    \EndIf
    \EndFor
    \end{algorithmic}
\end{algorithm}

As shown in an unpublished manuscript \cite{guo2022combinatorial}, DNA-MU also fails in terms of IC. Namely, a similar problem as in Example~\ref{exa:GIDM} exists also for DNA-MU. See Example~\ref{exa:DNA-MU}.

\begin{example}
\label{exa:DNA-MU}
This example is from \cite{guo2022combinatorial}. Let us run DNA-MU on the social network from Figure \ref{fig:single}.
When all reports are truthfully
DNA-MU does not allocate an item to $a$ as $a$'s valuation $v_a=3$ is not larger than the 4-th highest valuations $v_e=4$. Hence, the winners are $b,c,d,e$. Suppose, on the other hand, $f$ misreports her connection $r'_f=\varnothing$. Then DNA-MU would allocate an item to $a$, thus the elements in the other subtree face increased competition. This will result in the winners $a,b,c,f$. As in Example~\ref{exa:GIDM}, $f$ is able to manipulate the outcome to her advantage.
The detailed run-through of these two cases are presented in Table~\ref{tab:DNA-MU} and Table~\ref{tab:DNA-MU-2}.
\end{example}

\begin{table}
    \begin{center}
    \caption{Running example of DNA-MU in Fig.~\ref{fig:single} with $m=4$ assuming all buyers are truthful. The winners are $b,c,d,e$.}
    \label{tab:DNA-MU}
    \begin{tabular}{|c|c|c|c|}
    \hline
    rd & $m'$ & agent  &  $\pi,p$ \\
    \hline
    1 & 4 & $b$ & $\pi_b=1$, $p_b=0$ \\
    \hline
    2 & 3 & $c$ & $\pi_c=1$, $p_c=0$\\
    \hline
    3 & 2 & $d$  & $\pi_d=1$, $p_d=5$\\
    \hline
    4 & 1 & $e$  & $\pi_e=1$, $p_e=3$ \\
    \hline
    \end{tabular}
    \end{center}
\end{table}

\begin{table}
    \begin{center}
    \caption{Running example of DNA-MU in Fig.~\ref{fig:single} with $m=4$ assuming all buyers are truthful. The winners are $a,b,c,f$.}
    \label{tab:DNA-MU-2}
    \begin{tabular}{|c|c|c|c|}
    \hline
    rd & $m'$ & agent &  $\pi,p$ \\
    \hline
    1 & 4 & $a$ &  $\pi_a=1$, $p_a=1$ \\
    \hline
    2 & 3 & $b$ &  $\pi_b=1$, $p_b=0$\\
    \hline
    3 & 2 & $c$ &  $\pi_c=1$, $p_c=0$\\
    \hline
    4 & 1 & $f$ &  $\pi_f=1$, $p_e=6$ \\
    \hline
    \end{tabular}
    \end{center}
\end{table}

\section{The LDM-Tree mechanism}\label{app:ldm-tree}

 The {\em LDM-Tree mechanism} \cite{liu2022diffusion}  utilises {\em local} information given by {\em layers} of the tree, where layer $L_i$ contains agents whose distance from the seller is $i$. The auction runs several rounds. In round $i$, the mechanism only involves nodes in $L_i$ and those nodes in $L_{i+1}$ that do not pose a potential competition to nodes in $L_i$. Example~\ref{exa:LDM-Tree} shows that for single-unit auction (i.e., $m=1$) the social welfare of LDM-Tree may be arbitrarily worse than that of MUDAN, one of the earliest diffusion auction baselines.

\smallskip

\noindent {\bf Detailed procedure description.} The LDM-Tree mechanism allocates items layer by layer.
When the mechanism is considering allocation of a certain layer, it guarantees the buyers report their neighbours truthfully by removing the buyers below layer $L_i$ who are potential competitors of layer $L_i$. These competitors are divided into two parts. One part is the buyers who diffuse information to potential winners, these buyers can misreport her valuation to take more items from the high-priority layer and get more rewards through resale. The other part contains those buyers who are potential winners. So for each agent $j\in L_i$, the first part is a set of buyers $C_j^P(\theta')$ that contains all children of $j$ who have a  child. Then for another part, in the remaining children of $j$, the buyers who are potential winners need to be removed. If only remove the top-$m$ ranked buyers in $j$’s remaining children according to their valuation report for the first item is not IC. If the agent in $C_j^P(\theta')$ has high valuation but does not report her neighbours, she will be removed in  top-$m$ buyers. The buyers to be replaced will compete
with the high-priority layers to get more items. So the mechanism needs to remove the buyers $C_j^W(\theta')$ which includes the top $m+\mu-|C_j^P(\theta')|$ ranked buyers in $j$’s remaining children where $\mu$ is the upper bound of $|C_i^P(\theta')|$. The children to be removed of agent $j$ is denoted by $R_j$ and the children be removed by layer $L_i$ is denoted by $R_{L_i}$. After removing all these competitors of layer $L_i$, the mechanism performs the optimal allocation in the remaining buyers. If an agent $j\in L_i$ can be allocated $k$ items, she pays them for $SW_{-D_i}-(SW_{-R_{L_i}}-(v_1+\ldots+v_k))$. If $j\in L_i$ cannot be allocated an item, she gets a reward $\SW_{-D_i}-\SW_{-R_{L_i})}$. Then the remaining items are passed to the next layer $L_{i+1}$. See Algorithm~\ref{alg:ldm-tree}.

\begin{example}\label{exa:LDM-Tree}
Let us run LDM-Tree on the same network as in Figure~\ref{fig:single} with $m=1$. In the first round, the mechanism considers only agents in layer 1, and allocates the item to $a$, whereas MUDAN allocates the item to the highest bidder $f$.
Note that this holds when $v_f$ is arbitrarily high.
%
\end{example}

\begin{algorithm}[ht!]
    \caption{LDM-Tree mechanism}
    \label{alg:ldm-tree}
    \begin{algorithmic}[1]
    \State Calculate each agent's layer
    \State Initialise the residual supply $m'\leftarrow m$
    \For{layer $\ell \leftarrow 1,2,\ldots,\ell^{\max}$}
    \State Remove the competitor set $R_\ell$ and calculate the optimal allocation in the remaining buyers denoted by $\pi_i^\ell$
    \For{agent $i\in \ell$}
    \If{$\pi_i^\ell>0$}
    \State Allocate $\pi_i^\ell$ items to $i$ with payment  $$\SW_{-D_i}-(\SW_{-R_{L_i}}-(v_1+\ldots+v_k))$$
    \State Update $m' \leftarrow m'-\pi_i^\ell$
    \Else
    \State Reward $i$ by $\SW_{-D_i}-\SW_{-R_{L_i}}$
    \EndIf
    \EndFor
    \If{$m'=0$}
    \State Break
    \EndIf
    \EndFor
    \end{algorithmic}
\end{algorithm}

\section{The MUDAN mechanism}\label{app:mudan}

\begin{figure}[!htbp]
    \centering
   \includegraphics[width=0.5\columnwidth]{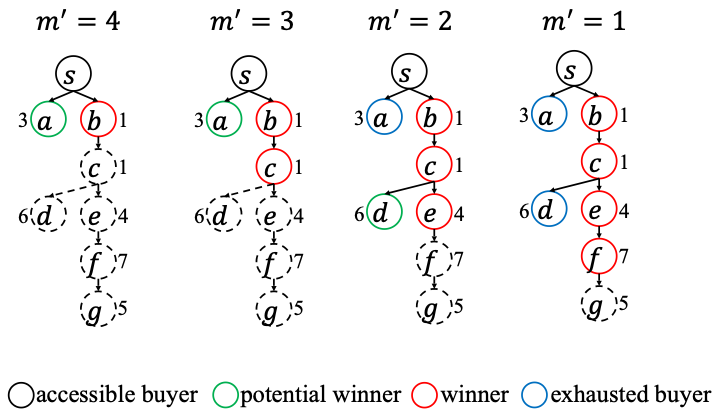}
    \caption{Detailed run-through of MUDAN on the network shown in Figure~\ref{fig:single}
    }
    \label{fig:mudan-detail}
\end{figure}

\medskip

\noindent {\bf Lemma \ref{lem:MUDAN-IR}} {\em  MUDAN satisfies IR, ND, and NW. }

\medskip

\begin{proof}
{\bf IR:} By our payment rule, suppose $w$ is selected as a winner at an iteration. The payment $p_w=v'_{m'+1}$ where $v'_{m'+1}$ is the $(m'+1)$th highest valuation in $A\setminus W$ while $v'_w$ is among the top-$m'$ valuations in $A\setminus W$. If $w$ reports the true valuation, then $v_w =v'_w \geq v_{m'+1}'$, which ensures that the utility $u_w= v_w - p_w \geq v_{m'+1}-v'_{m'+1}=0$. The utility $u_i$ of any other buyer $i \in A \backslash \{w\}$ is $0$.

\smallskip

{\bf ND:} This condition trivially holds as no buyer will have a negative payment (i.e., receiving a reward) in MUDAN.  

\smallskip

{\bf NW:} When the termination condition is met,  $W$ either contains exactly $m$ buyers, or less than $m$ buyers while all explored buyers are winners. Thus NW easily holds.
\end{proof}

\medskip

\noindent {\bf Theorem~\ref{thm:MUDAN}}
{\em MUDAN terminates within time $O(n^2+|E|)$, satisfies IC, IR, ND, NW, and $1/m$-weak efficiency.}

\medskip

\begin{proof}
We only need to prove the time complexity of MUDAN. The other properties are handled by lemmas above.

{\bf Termination.} As described above (in the main paper), the algorithm terminates in exactly $m$ iterations.

{\bf Complexity.} For running time analysis, we describe an implementation of the algorithm: Maintain a sorted list of buyers in $A$.
The buyers in this list are sorted in descending order of their valuations. In this way, the set $\PWin$ can be accessed in constant time. When a node is added in $A$, scan the list and add the node to the appropriate location. With each buyer in $A$, put a label indicating whether the node has been chosen as a winner, i.e., added in $W$. At each iteration, scan the list to find the node in $\PWin\setminus W$ with the highest priority and label it as a winner.

\begin{itemize}
    \item To explore the graph involves potentially checking all edges and nodes. This takes $O(n+|E|)$.
    \item To add all nodes in $A$ involves performing an insertion sort to the set of buyers. This takes $O(n^2)$.
    \item To choose $m$ winners from $\PWin$ involves $m$ iterations, each iteration scans the current list once. This takes another $O(n^2)$.
\end{itemize}
Therefore in summary the algorithm runs in time $O(n^2+|E|)$.
\end{proof}




\section{The MUDAR mechanism}\label{app:mudar}
In this section, we present our MUDAR ({\em Multi-Unit Diffusion Auction with Reward}) mechanism. 
MUDAR also implements the generic mechanism (Alg.~\ref{alg:framework}). The difference between MUDAN lies in how they incentivise winners. Once a winner $w$ is chosen, the MUDAN mechanism commits to allocating $w$ an item, thereby incentiving $w$ to diffuse the auction information. In this way, one may view the allocation result as being determined ``on the fly'' during the graph exploration.
On the other hand,  MUDAR may either allocate an item to $w$ or give $w$ a reward (i.e., a negative payment), which equals to the utility of $w$ had she been allocated an item. The allocation result is determined after the graph exploration is completed, when the buyers' connections are fully revealed.  In this way, MUDAR can identify  buyers that report the $m$-highest valuations globally.  At a given iteration:

\begin{itemize}[leftmargin=*]
    \item {\bf Potential winner set $\PWin$}: A locally accessible buyer $i$ is a {\em potential winner} if her valuation $v'_i$ is among the top-$m$ valuations in $A$.

    \item {\bf Tentative payment $\hat{p}_w$}: For the winner $w$ selected at the given iteration, $\hat{p}_w$ records how much $w \in W$ should pay 
    if she is allocated an item. Set $\hat{p}_w\coloneqq v'_{m+1}$.
\end{itemize}

The exploration terminates when the explored buyers with the top-$m$ valuations are all chosen as winners. This happens only when all nodes are explored (either in $W$ or exhausted). The mechanism then partitions the winner set $W$ into two subsets $\AW$ and $\RW$:
\begin{enumerate}[leftmargin=*]
    \item $\AW$ contains the buyers that have top-$m$ valuations. Each winner $w\in \AW$ is allocated an item. Set $\pi_w\coloneqq 1$ and $p_w\coloneqq \hat{p}_w$.
    \item $\RW$ contains the winners $W\setminus \AW$. They only get a reward but no item. If  $w\in \RW$, the mechanism sets $\pi_w\coloneqq 0$ and gives a reward which equals to her utility had she obtained the item, i.e., $p_w\coloneqq \hat{p}_w - v'_w < 0$.
\end{enumerate}
A run-through example of the MUDAR mechanism is provided in Table~\ref{tab:mudar} and Figure~\ref{fig:mudar-detail}.

\begin{table}
    \begin{center}   
    \caption{Running MUDAR on the network in Fig.~\ref{fig:single} with $m=4$ assuming all buyers report truthfully.  We give priorities to agents by the degree where node with higher degree gets higher priority. `Iter.'  `Incr. to $A$' , `$W$', and `$p$' indicate the iteration number. the nodes to be added to $A$ in each iteration,  winners in descending order of $v_i$, and the payment of the winner in the iteration when she is allocated an item,   respectively. The winners are $b,c,e,f,d,g$ with $\AW=\{d,e,f,g\}$ \& $\RW=\{b,c\}$.}
    \label{tab:mudar}
    \begin{tabular}{|c|c|c|c|c|}
    \hline
    Iter. & Incr. to $A$ & $P$ &  $W$  & $p$ \\
    \hline
    1 & $a,b $ & $\{a,b\}$ & $\{b\}$ &  $p_b=-1$\\
    \hline
    2 & $c$ &$\{a,b,c\}$ & $\{b,c\}$ &  $p_c=-1$\\
    \hline
    3 & $d,e$ &$\{a,c,d,e\}$ & $\{b,c,e\}$ &  $p_e=1$\\
    \hline
    4 & $f$ & $\{a,d,e,f\}$ & $\{b,c,e,f\}$  & $p_f=1$\\
    \hline
    5 & $g$ & $\{d,e,f,g\}$ & $\{b,c,e,f,d\}$ &  $p_d=3$\\
    \hline
    6 & $\varnothing$ & $\{d,e,f,g\}$ & $\{b,c,,e,f,d,g\}$ &$p_g=3$\\
    \hline
    \end{tabular}
    \end{center}
\end{table}

\begin{figure}[!htbp]
    \centering
   \includegraphics[width=0.5\columnwidth]{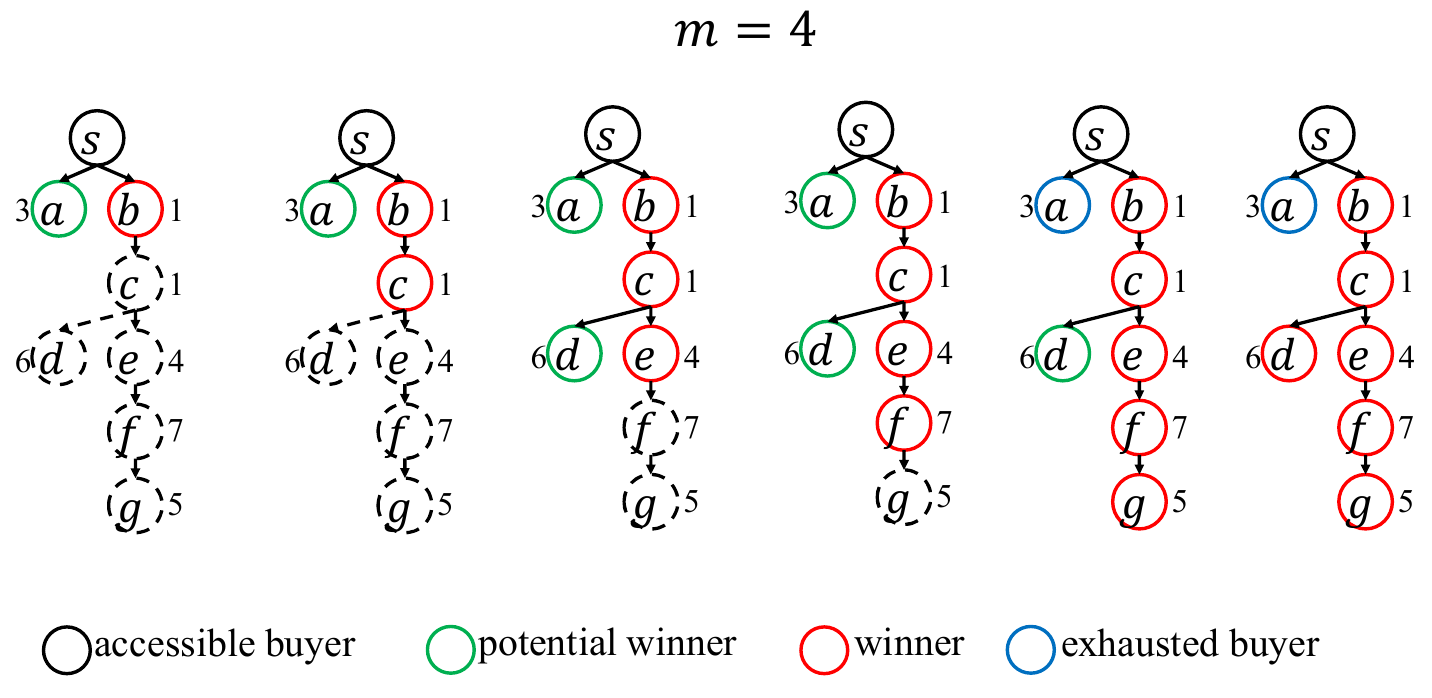}
    \caption{Detailed run-through of MUDAR on the network shown in Figure~\ref{fig:single}
    }
    \label{fig:mudar-detail}
\end{figure}

The next lemma shows that MUDAR has a number of desirable properties. 

\begin{lemma}\label{lem:MUDAR-property}
   MUDAR satisfies IR, ND, and NW.
\end{lemma}

\begin{proof}
{\bf IR:} By our payment rule, for winner $w\in \AW$,  $p_w=v'_{m+1}$ and thus $w$'s utility is $u_w=v_w - p_w=v_w-v'_{m+1}$; for winner $w \in \RW$, the utility $u_w=-p_w = v'_w - v'_{m+1}$. If she reports the true valuation, then $v_w =v'_w \geq v'_{m+1}$, which ensures $u_w\geq 0$. For any other agent $i \in A \setminus W$, $u_i=0$.

\medskip

{\bf NW:} When $|V_{\theta'}-1| \leq m$, the mechanism puts all buyers into set $\AW$ and $\sum_{i \in V_{\theta'} \setminus \{s\}} \pi_i(\theta') = |V_{\theta'}|-1$. When $|V_{\theta'}-1| > m$, the mechanism puts the $m$ buyers with the top-$m$ valuations in $\AW$ and $\sum_{i \in V_{\theta'} \setminus \{s\}}\pi_i(\theta') = m$.

\medskip

{\bf ND:} List all winners in the set $W$ as $w_1,w_2,\ldots, w_{|W|}$ in the order as they are added in $W$. Below we separately discuss $w_k$ for $k\leq m$ and for $k>m$:
\begin{enumerate}[leftmargin=*]
    \item $1\leq k\leq m$: The tentative payment $\hat{p}_{w_k}\geq 0$. If $w_k\in \RW$, the payment $p_{w_k}=\hat{p}_{w_k}-v_{w_k}\geq 0-v_{w_k}$; if $w_k \in \AW$,  $p_{w_k}\geq 0$.
    \item $m< k\leq |W|$: The tentative payment $\hat{p}_{w_k}$ is the $(m+1)$th highest valuation in $A$ at  iteration $k$. Let $\psi_k$ be the $(m+1)$th highest valuation in the set $W$ at iteration $k$. Since $W\subseteq A$, $\hat{p}_{w_k}\geq \psi_k$. If $w_k \in \RW$, the payment is $p_{w_k}=\hat{p}_{w_k}-v_{w_k} \geq \psi_{k}-v_{w_k}$; if agent $w_k\in \AW$, the payment is $p_{w_k}\geq \psi_{k}$.
\end{enumerate}
Summarising the above, the revenue is
\[
\RV(\theta')=\sum_{k=1}^{|W|}p_{w_k}
\]
which equals to
\begin{align*}
&\sum_{k=1}^m \{\hat{p}_{w_k}-v_{w_k}\mid w_k\in \RW\}+ \sum_{k=1}^m \{\hat{p}_{w_k}\mid w_k\in \AW\} + \\
& \sum_{k=m+1}^{|W|}\{\hat{p}_{w_k}-v_{w_k}\mid w_k\in \RW\}+\sum_{k=m+1}^{|W|}\{\hat{p}_{w_k}\mid w_k\in \AW\}\\
\geq &\sum_{k=1}^m \{0-v_{w_k}\mid w_k\in \RW\} +  0 + \\
&\sum_{k=m+1}^{|W|}\{\psi_{k}-v_{w_k}\mid w_k\in \RW\} +\sum_{k=m+1}^{|W|}\{\psi_{k}\mid w_k\in \AW\} \\
= &\sum_{k=m+1}^{|W|}  \psi_{k} - \sum_{w_k \in \RW} v_{w_k} = 0
\end{align*}
The last equation holds because $\psi_k$, $m\leq  k \leq |W|$, coincide with the valuations of $w_k\in \RW$, $1 \leq  k\leq |W|$.
\end{proof}

\paragraph*{\bf Truthfulness.}  MUDAR does not satisfy IC. Indeed, a winner in $\RW$ may report a higher value than the true valuation to receive a higher reward. Nevertheless, we now prove that MUDAR guarantees a weakened form of truthfulness. 

\begin{definition}\label{def:bounded ic}  
For a positive value $\mu>0$, a mechanism $\M$ is {\em $\mu$-bounded incentive compatible} ($\mu$-IC) if no buyer  $i\in B$ can benefit from  either  reporting a lower valuation than the true valuation, or reporting a valuation that is higher than $\mu$: for any $i\in B$, for any global profiles $\theta'$ and $\theta''$ such that $v_i''\geq \mu \geq v_i$ or $v_i''<v_i$,
\[
u_i(\theta_i,\theta_{-i}')\geq u_i(\theta_i'',\theta_{-i}'').\footnote{Similar to the def. of IC in Def.~\ref{def:ic}, $\theta''_{-i}$ is obtained from $\theta'_{-i}$.}
\]
\end{definition}

{\bf Remark.} The notion of $\mu$-bounded incentive compatibility resembles the notion of {\em one-way truthfulness} \cite{goel2014mechanism}, which asserts that no buyer will benefit from bidding lower than the true valuation. Definition~\ref{def:bounded ic} is stronger than one-way truthfulness as in $\mu$-bounded IC, the range of possible valuations at which a buyer $i$ may benefit from misreporting is a bounded interval $[v_i, \mu)$. In particular, if $v_i\geq \mu$, then $\mu$-bounded IC asserts that buyer $i$ will report the true profile.

Let $\mu$ be the $m$th highest valuation among all buyers. 
\begin{lemma}\label{lem:MUDAR-IC}
Suppose that the $m$th highest valuation $\mu$ among all buyers is public. Then MUDAR  achieves $\mu$-IC.
\end{lemma}

\begin{proof}
We first prove that there is no incentive for $w$ to misreport a lower valuation or a higher valuation $>\mu$.
If $w\in \AW$, truthful reporting of her valuation is $w$'s dominant strategy. This can be proved in a similar way as the IC property in Theorem~\ref{lem:MUDAN-IC}.
Then, for $w \in A\backslash W_A $. Consider the three alternative strategies of $w$:
\begin{enumerate}[leftmargin=*]
\item If $w$ reports $v_w' < v_{m+1}' \leq v_w$, then she is exhausted in an iteration and her utility will be $0$.
\item If $w$ reports $v_{m+1}' \leq v_w'<v_w$, then she will receive a lower utility, i.e., $v_w'-v_{m+1}'<v_w-v_{m+1}'$.
\item If she reports $v_{m+1} < v_w < \mu \leq v_w'$, her reported valuation will then be among the top-$m$ valuations. In this case, it must be that $w\in \AW$ and her utility becomes $v_w-v_{m+1}'$. This equals to the utility of truthful reporting.
\end{enumerate}
Then, for $i \in A\backslash W $, consider the three alternative strategies of $w$:
\begin{enumerate}[leftmargin=*]
\item If $i$ reports $v_i' \leq v_i$, then she is still exhausted in an iteration and her utility is still $0$.
\item If she reports $v_i\leq v_{m+1}< \mu \leq v_i'$, her reported valuation will then be among the top-$m$ valuations. In this case, it must be that $w\in \AW$ and her utility becomes $v_i-v_{m+1}'\leq 0$. This makes her a non-positive utility.
\end{enumerate}

Last, using essentially the same proof as that for MUDAN in Lemma~\ref{lem:MUDAN-IC}, one can prove that no buyer can benefit from misreporting her connections.
\end{proof}

\paragraph*{\bf Social welfare.} The algorithm of MUDAR terminates when all buyers are explored. Under the allocation and payment rules of MUDAR, the buyers form an equilibrium such that applying MUDAR leads to an efficient allocation. 


\begin{theorem}\label{thm:MUDAR}
Suppose that the $m$th highest valuation $\mu$ among all buyers is public. Then MUDAR terminates within time $O(n^2+|E|)$, satisfies $\mu$-IC, IR, ND, NW, and efficiency. 
\end{theorem}

\begin{proof}
{\bf Complexity.} The time complexity of the algorithm can be proved in the same way as for Theorem~\ref{thm:MUDAN}. We also maintain a sorted list to store the explored nodes whose first $m$ elements are those in $\PWin$, and label a node if it is added in $W$. In this implementation, the algorithm runs in $O(n^2+|E|)$.

%

{\bf Efficiency.} For the efficiency condition, note that when the termination condition is met, all buyers are explored. Next we show that in MUDAR these buyers form an equilibrium, which is efficient.

The proof of Lemma~\ref{lem:MUDAR-IC} has shown that given the value $\mu$, for each agent $w$ with valuation $v_i\geq \mu$, truthful reporting $v_w'=v_w$ is the dominant strategy.

While for each agent $i$ with valuation $v_i<\mu$, there are three cases:
\begin{enumerate}[leftmargin=*]
\item If $i$ reports a valuation such that she losses the auction, her utility is $0$.
\item If $i$ reports a valuation $v_i'>\mu$ such that she is an allocation winner, her utility is $u_i=v_i-\hat{p}_w\leq 0$.
\item If $i$ reports a valuation such that she is an reward winner, her utility is $u_i=v_i-(\hat{p}_w-v_i')\geq 0$. Further, $u_i$ is maximised when $v_i'=\mu$.
Therefore, reporting a valuation $v_i'=\mu$ is $i$'s dominant strategy.
\end{enumerate}
Then an equilibrium is formed where each agent $w$ with valuation $v_i\geq \mu$ truthfully reports and each other agent with valuation $v_i< \mu$ reports $v_i'=\mu$. In this equilibrium, MUDAR can selects the buyers with the top-$m$ true valuations as the winners, which shows the efficiency of MUDAR.  

\end{proof}

\section{Multi-demand multi-unit auction}\label{app:multi}

\paragraph*{\bf The MUDAN-$m$ mechanism.} 
Here, we give the definitions, statements and proofs regarding MUDAN-$m$ we omit in the main paper.   

\begin{lemma}\label{lem:multi}
   The following statements hold:
   \begin{enumerate}[leftmargin=*]
       \item For each buyer $i\in B$, the utility of $i$ received from $\M$ equals the sum of utilities of $i_1,\ldots,i_m$ received from $\widetilde{\M}$, i.e., $u_i(\eta')=\sum_{j=1}^m u_{i_j}(\theta')$.
       \item The social welfare $\SW(\eta')$ of $\M$ equals the social welfare $\SW(\theta')$ of $\widetilde{\M}$.
       \item The optimal social welfare on $B$ equals the optimal social welfare on $\widetilde{B}$.
       \item The revenue $\RV(\eta')$ of $\M$ equals the revenue $\RV(\theta')$ of $\widetilde{\M}$.
       \item The number of items allocated by $\M$ equals the number of items allocated by $\widetilde{\M}$.   \QED
   \end{enumerate}
\end{lemma}

In the next definition, let $\eta'_{-i}\coloneqq (\eta'_1,\ldots,\eta'_{i-1},\eta'_{i+1}, \ldots,$ $\eta'_n)$ denote the profiles of all buyers but $i$.
\begin{definition} \label{def:ic-m}
Let $\M$ be a mechanism.
\begin{enumerate}[leftmargin=*]
\item $\M$ is {\em incentive compatible} (IC) if for any buyer reporting truthfully is a dominant strategy: for all $i\in B$, all global profiles $\eta'$ and $\eta''$, we have $u_i(\eta_i, \eta_{-i}')\geq u_i(\eta_i'',\eta_{-i}'')$ \footnote{As $t_i''$ may be different from $t_i'$, some agents $j$ who are reachable in $G_{\eta'}$ may become unreachable had we replace $\eta'_i$ with $\eta''_i$. $\eta''_{-i}$ is obtained from $\eta'_{-i}$ by replacing $\eta'_j$ with the silent profile for all such agents $j$.}.
\item $\M$ is {\em individually rational} (IR) if any buyer by reporting truthfully receives  non-negative utility, i.e., for all $i\in B$, $u_i(\eta)\geq 0$.
\item $\M$ is {\em non-deficit} (ND) if for any global profile $\eta'$, the revenue is non-negative, i.e., $\RV(\eta')\geq 0$.
\item $\M$ is {\em non-wasteful} (NW) if all items are allocated to buyers (up to the number of reachable buyers), i.e., for any global profile $\eta'$, $\sum_{i \in V_{\eta'} \backslash \{s\}} \sum_{j=1}^m \pi_{i,j}(\eta') = \min\{m, m(|V_{\eta'}|-1)\}$.
\item $\M$ is {\em efficient}  if it achieves optimal social welfare, i.e., for any $\eta'$, $\SW(\eta')=\SW_{\opt}$.
\end{enumerate}
\end{definition}

\medskip

\noindent {\bf Lemma~\ref{lem:mudan-m-ic}} {\it The MUDAN-$m$ mechanism is IC.}

\medskip

\begin{proof}
We first prove that no buyer in $B$ can benefit from misreporting her valuation vector.
To run $\M$, we first execute the MUDAN mechanism $\widetilde{\M}$ given the constructed profile $\theta'$ on $\widetilde{B}$. Now consider the execution of $\widetilde{\M}$.

If $i_1$ is selected as a winner, then $\widetilde{\M}$ will proceed to select more $i_j$ as winners until the valuation is so low that $i_j$ is not added in $\PWin$. That is because after $i_j$ is chosen as winner, the only buyer in $\widetilde{B}$ that is added to $A$ is $i_{j+1}$. By definition of the priority on $B$, if $i_{j+1}$ is added to $\PWin$, then $i_{j+1}$ will be the next node with the highest priority and is thus chosen  as the winner.

Consider a given iteration of $\widetilde{\M}$. By the argument above, if a buyer $i_j\in \PWin$ is not chosen as the winner, then $j=1$. Moreover, if $i_j$ is exhausted, then $i_{\ell}$ for all $\ell>j$ are exhausted also due to the diminishing valuation. This means that for any buyer $i\in B$, either $i_1\in \PWin$, or $\{i_1,\ldots,i_m\}\cap \PWin=\varnothing$. This observation is crucial for the proof of truthfulness.

Now suppose in the given iteration, $w_j$ is chosen as the winner. Suppose $w\in B$ misreports her valuation vector so that $\vec{v'_w}\neq \vec{v_w}$. We show that this strategy will not give $w$ extra utility at this iteration:
\begin{enumerate}[leftmargin=0.6cm]
    \item[(a)] Suppose $v_{w,\ell}\neq v'_{w,\ell}$ where $\ell>j$. At this iteration, $w_\ell$ would {\em not} have been added in $A$. Therefore for $w$, misreporting $v_{w,\ell}$ does not affect the tentative payment $\hat{p}(w_j)$ for her $j$th item. 
    \item[(b)] Suppose $v_{w,k}\neq v'_{w,k}$ where $k<j$. Note that $w_k$ must have been a winner in an earlier iteration. So the valuation of $w_k$ is not taken into consideration when the algorithm sets the tentative payment $\hat{p}(w_j)$. 
    \item[(c)] Now suppose $w$ misreports the valuation of $v_{w,j}$ in such a way that $v'_{w,j}>v_{w,j}>\hat{p_{w_j}}$ or $v_{w,j}>v'_{w,j}>\hat{p_{w_j}}$. Here, $\hat{p_{w_j}}$ is the tentative payment of $w_j$ assuming $\tilde{\M}$ sets $w_j$ as the winner. As this does not change $w_j$'s priority, $w$ is still chosen as the winner in this iteration and is allocated her $j$th item with payment $\hat{p_{w_j}}$. This item will add $w$ a utility of $u_{w_j}((v'_{w_j},r_{w_j}), \theta_{-w_j})=u_{w_j}((v_{w_j},r_{w_j}), \theta_{-w_j})$.
    \item[(d)] On the other hand, if  $v'_{w,j}<\hat{p_{w_j}}<v_{w,j}$, then $w$ loses the $j$th item as all $w_{\ell}$ where $\ell\geq j$ are exhausted, giving her no extra utility.
\end{enumerate}
Summarising (a)-(d) above, in the given iteration, misreporting any element in valuation vector $(v_{w,1},\ldots,v_{w,m})$ will not give $w$ extra utility.

\smallskip

Next, we prove that for buyer $i\neq w$ where $i_1\in \PWin$, $i$ cannot benefit from misreporting her valuation using a similar argument:
    \begin{enumerate}[leftmargin=0.6cm]
        \item[(a)] Consider $j>1$. The node $i_j$ would not have been added in $A$. Therefore for $i$, misreporting $v_{i,j}$ cannot lead to a higher priority nor a less payment.
        \item[(b)] Suppose $i$ misreports her valuation such that $v'_{i_1}<v'_{m'+1}<v_{i_1}$. Recall that we used $v'_{m'+1}$ in MUDAN to denote the $m'+1$th highest valuation in this iteration, where $m'$ has the same meaning as defined in Sec.~\ref{sec:mudan}. Then $i_1$ would not be a potential winner, but rather be exhausted. This means that no item would be allocated to $i$ and her utility would be $0$.
        \item[(c)] Suppose $i$ misreports her valuation such that $v'_{i_1}>v_{i_1}>v'_{m'+1}$ or $v_{i_1}>v'_{i_1}>v'_{m'+1}$. Then $i_1$ would still be a potential winner in this iteration, giving her no extra utility.
    \end{enumerate}
Summarising (a)-(c) above, in the given iteration, misreporting any (combination of) elements in the valuation vector $(v_{i,1},\ldots,v_{i,m})$ will not give $i$ extra utility.

\smallskip

Last, consider buyer $i\neq w$ where $\{i_1,\ldots,i_m\}\cap \PWin=\varnothing$. We prove that $i$ cannot benefit from misreporting her valuation again using a similar argument:
    \begin{enumerate}[leftmargin=0.6cm]
        \item[(a)] Consider $j>1$. The node $i_j$ would not have been added in $A$. Therefore for $i$, misreporting $v_{i_j}$ cannot lead to a higher priority nor a less payment.
        \item[(b)] Suppose $i$ misreports her valuation such that  $v_{i_1} < v'_{m'} \leq v_{i_1}'$. Then we have: (i) If $i_1$ has the highest priority, then her utility of this item is $u_{i_1}((v_{i_1}',r_{i_1}),\theta_{-i_1})=v_{i_1}-v'_{m'+1} < 0 = u_{i_1}((v_{i_1},r_{i_1}),\theta_{-i_1})$. (ii) Otherwise, her utility remains $0$.
        \item[(c)] Suppose $i$ misreports her valuation such that $v_{i_1}<v_{i_1}'<v'_{m'}$ or $v'_{i_1}<v_{i_1} < v'_{m'}$. Her utility remains $0$ and all $i_j$ where $j\geq 1$ are exhausted.
    \end{enumerate}
Summarising (a)-(c) above, in the given iteration, misreporting any element in the valuation vector $(v_{i,1},\ldots,v_{i,m})$ will not give $i$ extra utility. This we proved the fact for all iterations, no buyer in $B$ has incentive to misreport their valuation vector.


It remains to prove that no buyer in $B$ can benefit from misreporting her neighbour set. Our argument is the following: For each agent $i_j\in \widetilde{B}$, her priority cannot increase when she hide any of her neighbours. And her neighbours is only added in $A$ either when $i_m$ is chosen as a winner or when $i_m$ is exhausted so that her neighbours cannot influence her allocation. The rest of the proof is the same as in the proof of Lemma~\ref{lem:MUDAN-IC}.
\end{proof}

For analysing the social welfare of MUDAN-$m$, we make the following definition:
\begin{itemize}
    \item Suppose $w^*_j\in \widetilde{B}$ is the winner chosen by the MUDAN mechanism $\hat{\M}$ in the last iteration. We say that $w^*$ is {\em critical} for a buyer $i$ if all paths from $s$ to $i$ in $G_{\eta'}$ pass through $w^*$.
\item Let $B^*\subseteq B$ denote the set of buyers for whom $w^*$ is not critical.
\item The {\em weakly-optimal social welfare} $\SW_{\wopt}$ denote the sum of the top-$m$ valuations among $\{v_{i,j}\}_{i\in B^*, 1\leq j\leq m}$.
\item  A mechanism $\M$ is {\em $\epsilon$-weakly efficient} if for any global profile $\eta'$, we have $\SW(\eta') \geq \epsilon \SW_{\wopt}$.
\end{itemize}

\medskip

\noindent {\bf Theorem~\ref{thm:mudan-m}} {\em  MUDAN-$m$ is IC, IR, ND, NW, and $1/m$-weakly efficient.}

\medskip

\begin{proof} IC follows from Lemma~\ref{lem:mudan-m-ic}, IR, ND, and NW follow from Lemma~\ref{lem:multi-property} and Theorem~\ref{thm:MUDAN}.
 $1/m$-weak efficiency follows directly from the fact that $\widetilde{\M}$ selects the buyer $i_1$ who has the highest valuation $v_{i_1}$ among the explored buyers.
 \end{proof}

\paragraph*{\bf The MUDAR-$m$ mechanism.} For analysing the truthfulness of MUDAR-$m$, we make the following definition:

For a positive value $\mu>0$, a mechanism $\M$ is {\em $\mu$-bounded incentive compatible} ($\mu$-IC) if for any buyer  $i\in B$, for any $1\leq j\leq m$, $i$ cannot benefit from either reporting a lower valuation for her $j$th item, or reporting a valuation that is higher than $\mu$, i.e., for any $i\in B$, for any $1\leq j\leq m$, for any global profiles $\eta'$ and $\eta''$ such that $v_{i,j}''\geq \mu \geq v_{i,j}$ or $v_{i,j}''<v_{i,j}$,
\[
u_i(\eta_i,\eta_{-i}')\geq u_i(\eta_i'',\eta_{-i}'').\footnote{Similar to the def. of IC in Def.~\ref{def:ic}, $\eta''_{-i}$ is obtained from $\eta'_{-i}$.}
\]

We first show that MUDAR-$m$ satisfies the same truthfulness condition as MUDAR.

\begin{lemma}\label{lem:mudar-m-ic}
Suppose that the $m$th highest valuation $\mu$ among $v_{i,j}$ for all $1\leq i\leq n$, $1\leq j\leq m$ is public. Then MUDAR-$m$ mechanism satisfies $\mu$-IC. 
\end{lemma}

\begin{proof}
    Firstly we prove that no agent $i\in B$ can benefit from misreporting her valuation vector in such a way that some of its elements $v'_{i,j}$ are lower than the true valuation $v_{i,j}$, or higher than $\mu$. In other words, suppose $v'_{i,j}<v_{i,j}$ or $v'_{i,j}>\mu$ for some $1\leq j\leq m$, then the utility that $i$ receives is strictly less than as if she reports $\vec{v_i}$ truthfully.

    Just like in the proof of Lemma~\ref{lem:mudan-m-ic}, in the following we analyse the execution of the $\widetilde{\M}$ mechanism over buyers $\widetilde{B}$. Consider a given iteration of $\widetilde{\M}$.

    Let $w_j$ be the chosen winner by $\widetilde{\M}$. Suppose $w_j$ is selected as an allocation winner. Suppose $w\in B$ misreports her valuation vector so that $\vec{v'_w}\neq \vec{v_w}$. We show that this strategy will not give $w$ extra utility at this iteration:
    \begin{enumerate}[leftmargin=0.6cm]
        \item[(a)] Suppose $v_{w,\ell}\neq v'_{w,\ell}$ where $\ell>j$. In this iteration, $w_l$ would not have been added in $A$. Therefore for $w$, misreporting $v_{w,\ell}$ does not affect the tentative payment $\hat{p}(w_j)$.
        \item[(b)] Suppose $v_{w,k}\neq v'_{w,k}$ where $k<j$. Note that $w_k$ must have been an allocation winner in an earlier iteration. We have $v'_{w,k}>v'_{w,j}>\hat{p}(w_j)$, so that $v'_{w,k}$ cannot be tentative payment of $w_j$. Therefore for $w$, misreporting $v_{w,\ell}$ does not affect the tentative payment $\hat{p}(w_j)$.
        \item[(c)] Now suppose $w$ misreports the valuation of $v_{w,j}$ in such a way that $v'_{w_j}>v_{w_j}>\hat{p}(w_j)$ or $v_{w_j}>v'_{w_j}>\mu$, where $\hat{p}(w_j)$ is the tentative payment of $w_j$. As this does not change $w_j$'s priority, $w$ is still chosen as the winner in this iteration and is allocated her $j$th item and with payment $\hat{p}(w_j)$. This item will add $w$ a utility of $u_{w_j}((v'_{w_j},r_{w_j}),\theta_{-w_j})=u_{w,j}((v'_{w,j},r_{w_j}),\theta_{-w_j})$.
        \item[(d)] Suppose $w$ misreports her valuation such that  $v_{w_j}>\mu>v'_{w_j}>\hat{p}(w_j)$, then $w_j$ becomes a reward winner in this iteration. This iteration gives $w$ a lower extra utility of $u_{w_j}((v'_{w_j},r_{w_j}),\theta_{-w_j})=v'_{i,j}-\hat{p}(w_j)\leq v_{i,j}-\hat{p}(w_j)=
        u_{w_j}((v'_{w,j},r_{w_j}),\theta_{-w_j})$.
        \item[(e)] On the other hand, if $w$ misreports her valuation such that $v'_{w_j}<\hat{p}(w_j)<v_{w,j}$, then $w$ loses her $j$th item as all $w_\ell$ where $\ell>j$ are exhausted, giving her no extra utility.
    \end{enumerate}

    Summarising (a)-(e) above, in the given iteration, misreporting any element in the valuation vector $(v_{w,1},\ldots,v_{w,m})$ will not give $w$ extra utility.

    Next, consider buyer $w_j$ is selected as a reward winner. We prove that $w$ cannot benefit from misreporting any element of her valuation vector such that $v'_{w,j}<v_{w,j}$ or $v'_{w,j}>\mu$:
    \begin{enumerate}[leftmargin=0.6cm]
        \item[(a)] Suppose$v_{w,\ell}\neq v'_{w,\ell}$ where $\ell>j$. In this iteration, $w_l$ would not have been added in $A$. Therefore for $w$, misreporting $v_{w,\ell}$ does not affect the tentative payment $\hat{p}(w_j)$.
        \item[(b)] Suppose $v_{w,k}\neq v'_{w,k}$ where $k<j$. Note that $w_k$ must have been a winner in an earlier iteration. We have $v'_{w,k}>v'_{w,j}>\hat{p}(w_j)$, so that $v'_{w,k}$ cannot be tentative payment of $w_j$. Therefore for $w$, misreporting $v_{w,\ell}$ does not affect the tentative payment $\hat{p}(w_j)$.
        \item[(c)] Now suppose $w$ misreports the valuation of $v_{w,j}$ in such a way that $v'_{w_j}>\mu>v_{w_j}>\hat{p}(w_j)$, then $w_j$ is selected as an allocation winner. This iteration gives $w$ an extra utility of $v_{w,j}-\hat{p}(w_j)$. This equals to the utility of truthful reporting.
        \item[(d)] On the other hand, if $w$ misreports her valuation such that $v'_{w_j}<\hat{p}(w_j)<v_{w_j}$, then $w$ loses her reward of the $j$'s item as all $w_\ell$ where $\ell>j$ are exhausted, giving her no extra utility.
    \end{enumerate}

    Summarising (a)-(d) above, in the given iteration, misreporting any element in the valuation vector $(v_{w,1},\ldots,v_{w,m})$ will not give $w$ extra utility.
   Next, consider buyer $i\neq w$ where $i_1\in \PWin$. We prove that $i$ cannot benefit from misreporting her valuation using a similar argument as above:
    \begin{enumerate}[leftmargin=0.6cm]
        \item[(a)] Consider $j>1$. The node $i_j$ would not have been added in $A$. Therefore for $i$, misreporting $v_{i,j}$ cannot lead to a higher priority nor a less payment.
        \item[(b)] Suppose $i$ misreports her valuation such that $v'_{i_1}<v'_{m'+1}<v_{i_1}$. Then $i_1$ would not be a potential winner, but rather be exhausted. This means that no item would be allocated to $i$ and her utility would be $0$.
        \item[(c)] Suppose $i$ misreports her valuation such that $v'_{i_1}>v_{i_1}>v'_{m'+1}$ or $v_{i_1}>v'_{i_1}>v'_{m'+1}$. Then $i_1$ would still be a potential winner in this iteration, giving her no extra utility.
    \end{enumerate}
Summarising (a)-(c) above, in the given iteration, misreporting any (combination of) elements in the valuation vector $(v_{i,1},\ldots,v_{i,m})$ will not give $i$ extra utility.

\smallskip

Last, consider buyer $i\neq w$ where $\{i_1,\ldots,i_m\}\cap \PWin=\varnothing$. We prove that $i$ cannot benefit from misreporting her valuation again using a similar argument:
    \begin{enumerate}[leftmargin=0.6cm]
        \item[(a)] Consider $j>1$. The node $i_j$ would not have been added in $A$. Therefore for $i$, misreporting $v_{i_j}$ cannot lead to a higher priority nor a less payment.
        \item[(b)] Suppose $i$ misreports her valuation such that  $v_{i_1} < v'_{m'} \leq v_{i_1}'$; Recall that we used $v_{m'}$ in MUDAN to denote the $m'$th highest valuation in this iteration, where $m'$ has the same meaning as defined in Sec.~\ref{sec:mudan}. Then we have: (i) If $i_1$ has the highest priority, then her utility of this item is $u_{i_1}((v_{i_1}',r_{i_1}),\theta_{-i_1})=v_{i_1}-v'_{m'+1} < 0 = u_{i_1}((v_{i_1},r_{i_1}),\theta_{-i_1})$. (ii) Otherwise, her utility remains $0$.
        \item[(c)] Suppose $i$ misreports her valuation such that $v_{i_1}<v_{i_1}'<v'_{m'}$ or $v'_{i_1}<v_{i_1} < v'_{m'}$. Her utility remains $0$ and all $i_j$ where $j\geq 1$ are exhausted.
    \end{enumerate}
Summarising (a)-(c) above, in the given iteration, misreporting any element in the valuation vector $(v_{i,1},\ldots,v_{i,m})$ will not give $i$ extra utility. This we proved the fact for all iterations, no buyer in $B$ has incentive to misreport their valuation vector.

\medskip

It remains to prove that no buyer in $B$ can benefit from misreporting her neighbour set. Our argument is the following: For each agent $i_j\in \widetilde{B}$, her priority cannot increase when she hide any of her neighbours. And her neighbours is only added in $A$ either when $i_m$ is chosen as a winner or when $i_m$ is exhausted so that her neighbours cannot influence her allocation. The rest of the proof is the same as in the proof of Lemma~\ref{lem:MUDAN-IC}.
\end{proof}

\begin{theorem}\label{thm:mudar-m}
Suppose that the $m$th highest valuation $\mu$ among $v_{i,j}$ for all $1\leq i\leq n$, $1\leq j\leq m$ is public. Then MUDAR-$m$ is $\mu$-IC, IR, ND, NW, and efficient. 
\end{theorem}

\begin{proof}
IR, ND, NW and efficiency follow directly from Lemma~\ref{lem:multi-property} and Theorem~\ref{thm:MUDAR}. $\mu$-IC follows from Lemma~\ref{lem:mudar-m-ic}.
\end{proof}

\end{document}